\documentclass[journal, twoside]{IEEEtran}

\usepackage[caption=false,font=normalsize,labelfont=sf,textfont=sf]{subfig}
\usepackage{stfloats}
\usepackage{url}
\usepackage{verbatim}
\usepackage{cite}
\usepackage{amsmath,amssymb,amsfonts}
\usepackage{graphicx}
\usepackage{textcomp}
\usepackage{xcolor}
\usepackage{booktabs}
\def\BibTeX{{\rm B\kern-.05em{\sc i\kern-.025em b}\kern-.08em
    T\kern-.1667em\lower.7ex\hbox{E}\kern-.125emX}}
\usepackage{tikz}
\usepackage{adjustbox}
\usepackage{booktabs}
\usepackage{colortbl}
\usepackage{multirow}
\usepackage{algorithm}
\usepackage{algpseudocode}
\usepackage{amsmath, amsthm}
\usepackage{float} 
\usepackage{pgfplots}
\pgfplotsset{compat=1.15}
\usetikzlibrary{patterns}
\DeclareMathOperator*{\argmax}{arg\,max}

\usepackage{tikz}
\usepackage{amsmath}
\usepackage{amsfonts}
\usepackage{filecontents}
\usepackage{bbding}
\usepackage{algorithm}
\usepackage{algpseudocode}
\usepackage{algorithmicx}

\usepackage[normalem]{ulem}
\usepackage{hyperref}
\usepackage{array}
\newtheorem{theorem}{Theorem}
\newtheorem{lemma}{Lemma}
\newtheorem{definition}{Definition}

\begin{document}
\title{Dependency-Aware Dominant Resource Fairness for Multi-Tenant Multi-Resource Systems 
}



\author{Zeidan~Braik, Francesca~Fossati, Sahar~Hoteit,~\IEEEmembership{Senior Member,~IEEE,} Stefano~Secci,~\IEEEmembership{Senior Member,~IEEE}
\thanks{Zeidan Braik is with the Conservatoire National des Arts et M\'etiers (Cnam), C\'edric, Paris, France.
(email: zeidan.braik@cnam.fr).}
\thanks{Francesca Fossati is with the Sorbonne Universit\'e, CNRS LIP6, Paris, France (email: francesca.fossati@sorbonne-universite.fr).}
\thanks{Sahar Hoteit is with the Universit\'e Paris-Saclay, CNRS L2S, CentraleSup\'elec, Gif-sur-Yvette, and with Institut Universitaire de France (IUF), France (email: sahar.hoteit@centralesupelec.fr).}
\thanks{Stefano Secci is with the Conservatoire National des Arts et M\'etiers (Cnam), C\'edric, Paris, France (email: stefano.secci@cnam.fr).}
\thanks{This work was funded by the French ANR HEIDIS project (contract nb: ANR-21-CE25-0019; \url{https://heidis.roc.cnam.fr}) and the IPCEI ME/CT Orange project (contract nb: DOS0239248/00; \url{https://piiec5g.wp.imt.fr}).}
\thanks{This work has been submitted to the IEEE for possible publication. Copyright may be transferred without notice, after which this version may no longer be accessible.}}

\markboth{IEEE Transactions on Networking,~VOL.~XX,~YYYY }%
{Braik \MakeLowercase{\textit{et al.}}: Dependency-Aware Dominant Resource Fairness for Multi-Tenant Multi-Resource Systems}

\maketitle


\begin{abstract}
Multi-resource allocation in network-congested, multi-tenant systems in which demand exceeds available capacity is challenging, as there is no straightforward way to determine how much of each resource to assign, especially when resources are interdependent. Classical approaches such as Dominant Resource Fairness (DRF), which generalizes Max-Min Fairness (MMF) to multiple resources, assume linear proportional dependencies across resources, requiring allocations to follow fixed proportions implied by tenants’ demands. However, this assumption may lead to inefficient allocations and resource waste, with allocated resources that go unused in practice. In this paper, we consider a multi-resource orchestrator and propose the Dependency-aware Dominant Resource Fairness (DDRF) policy, a centralized generalization of DRF  that considers inter-resource dependencies: it equalizes active dominant shares of congested resources, preserving DRF’s desirable properties, while avoiding its inefficiency with low-demand tenants. We prove that DDRF always saturates at least one congested resource, ensuring Pareto efficiency and eliminating resource waste. We evaluate DDRF using Amazon EC2 traces and a virtualized radio access network (vRAN) use case while considering real resource dependencies. The results show that DDRF improves effective user satisfaction by up to $80\%$ and reduces resource waste by up to $60\%$ compared to dependency-agnostic baselines, while improving Jain's fairness index by more than $15\%$ compared to the utilitarian  policy.
\end{abstract}


\vspace{-0.5 cm}
\section{Introduction}
\IEEEPARstart{I}{n} computing and networking, resource allocation has been a central research direction for decades, spanning contexts from scheduling in operating systems~\cite{chandra2000surplus} to resource management in datacenters, cloud and radio networks~\cite{boutin2014apollo}. More recently, its necessity has become especially prevalent in networking~\cite{slicing5G}, where multi-tenant\footnote{Hereafter the terms "tenants" and "users" are used interchangeably, we use the feminine pronoun ("she","her"); this choice is arbitrary.} infrastructures, such as 5G and cloud-native systems, must share finite resources across diverse services. A fundamental challenge in these environments is congestion, which arises when aggregate traffic exceeds the capacity of at least one resource. Under congestion, resource allocation must decide how to divide scarce capacity across demand vectors while adopting principled notions of fairness, ensuring tenants are treated equitably under contention.

The theory of single-resource allocation is largely complete~\cite{proportionalfairness, proportionalfairnessequation, bertsekasbook, mo2002fair, lan2010axiomatic}, with well-established fairness notions such as \textit{Max-Min Fairness (MMF)}~\cite{bertsekasbook} and \textit{Proportional Fairness (PF)}~\cite{proportionalfairness,proportionalfairnessequation}. In contrast, multi-resource allocation remains an active frontier~\cite{DRFpaper, wang2014multi, joe2013multiresource, MURANESPaper}, as recent efforts seek to extend single-resource fairness to settings with multiple resources~\cite{DRFpaper,MURANESPaper}. 

The extensions to multi-resource rules, developed in contexts ranging from cloud computing to congestion control, formalize different notions of fairness. However, multi-resource rules, throughout the literature, are often built on one of two simplifying assumptions: the first is a \textit{linear proportional dependency} between resources (e.g., if a user requests 2 units of resource 1 and 4 units of resource 2, then the allocation is expected to maintain the same 2:4 ratio across both resources); the second is a structured utility model, most notably the \textit{Leontief utility}~\cite{bonaldbmf, hug, surveyresourceallocation}, where a user's satisfaction is determined by the smallest proportion of resources received, an approach that, in effect, also reduces to linear proportional dependency. Under such linear assumptions, the inherently multi-dimensional problem is reduced to a single-resource problem, as the task becomes finding a single ratio per user that scales its resource demands uniformly.

\begin{figure}
\centering
    \includegraphics[width=\columnwidth]{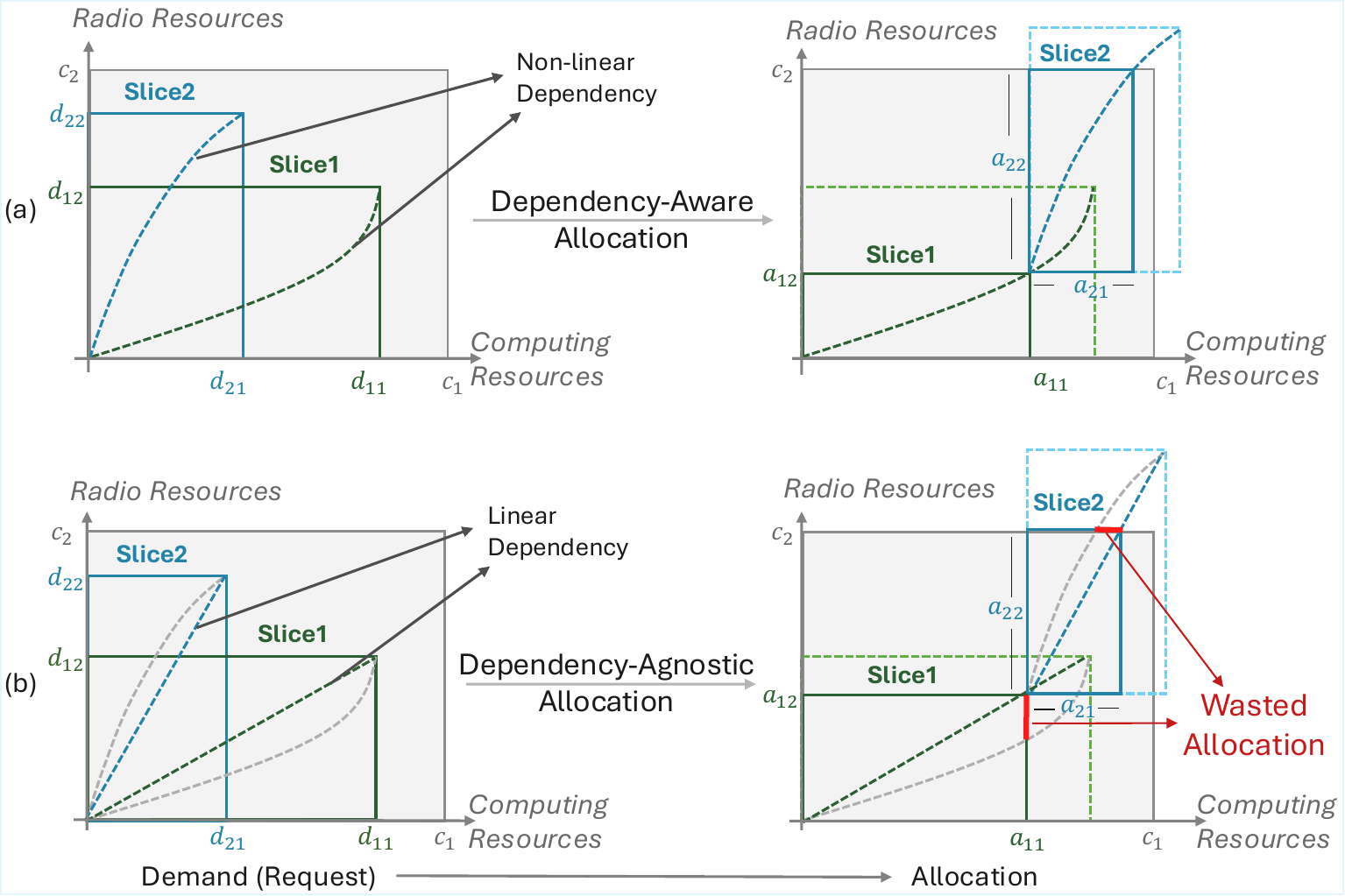}
    \vspace{-0.7cm}
    \caption{Comparison of dependency-aware and dependency-agnostic allocations for two slices. Left: demands with a nonlinear dependency in (a) and the imposed linear dependency in (b). Right: resulting allocations; enforcing the linear relation in (b) yields an unused portion, i.e., resource waste.}
    \label{fig:dependency}
\end{figure}

The need to account for \textit{generic and realistic dependencies} is increasingly critical in modern applications, particularly in networking. For instance, in mobile networks~\cite{polese2023understanding}, radio and computing resources are coupled in nonlinear ways~\cite{khatibi2018modelling, hojeij2025flexible}, moreover link bandwidth over midhaul and backhaul~\cite{xiao2020can} depends on the transport block size (TBS), which is defined in 3GPP specifications as a piecewise-affine function of the number of resource blocks~\cite{3gpp.38.214}. 
Existing state-of-the-art policies often oversimplify these dependencies e.g. by assuming linear ratios across resources. We refer to these policies in the following as \textit{Dependency-Agnostic} ones.
Ignoring realistic dependencies leads to allocation waste (i.e., resources allocated to tenants but could not be used in practice). Figure~\ref{fig:dependency} shows an example of two network slices: slice 1 asking for small amount of  radio resources but higher computing resources $(d_{11}$ and $d_{12})$, respectively  while slice 2 is asking for the opposite (these demands are denoted $d_{21}$, $d_{22}$). 
Top row shows dependency-aware allocations, the bottom row assumes fixed linear proportions (\emph{dependency-agnostic}). 
  The computing resources have a limited capacity, denoted by $c_1$
 while the radio resources have a limited capacity, denoted by $c_2$.
Rectangles on the right represent the allocated amount of resources ($a_{11}, a_{12}$ for slice 1 and $a_{21}, a_{22}$ for slice 2), with dashed outlines denoting original demands.
As shown in the figure, dependency-agnostic allocation leads to resource waste, i.e., resources are allocated but remain unusable in practice due to existing inter-resource dependencies and potential bottlenecks in one or more resources.

The presence of generic inter-resource dependencies has been identified in~\cite{MURANESPaper}; however, to the best of our knowledge, no other work in the literature has proposed a multi-resource allocation approach that explicitly accounts for these dependencies.
Authors in~\cite{MURANESPaper} introduce a generic formulation that avoids linear aggregation across resources by defining a decision variable for each user--resource pair, thereby enabling dependencies to be expressed directly as constraints. Nevertheless, that way does not provide a systematic characterization of dependency models, nor does it adapt the objective function to account for such dependencies. To fill this gap, we propose a dependency-aware allocation scheme that explicitly accounts for inter-resource dependencies, with the aim of preventing resource waste while ensuring a fair distribution of resources.

To design a fair allocation mechanism, two main approaches exist: (i) maximize a Schur-concave utility function~\cite{marshall1979inequalities,mo2002fair,MURANESPaper} and then analyze the fairness properties induced by the optimizer, or (ii) start from a target fairness principle and impose it explicitly as a constraint while maximizing utilization (i.e., satisfied demand)~\cite{zaharia2010delay,DRFpaper,popa2012faircloud}. In this work, we adopt the second approach. Specifically, we generalize \textit{max-min fairness} to the dependency-aware multi-resource setting, where equity is imposed directly on allocations: tenants receive equal allocations unless a tenant's demand is smaller, in which case it is fully satisfied. This corresponds to the \textit{leximin principle}, inspired by Rawlsian fairness~\cite{rawls1971theory}, which maximizes the allocation of the least-served tenant before considering the next, and so on. Our approach differs from the Dominant Resource Fairness (DRF)~\cite{DRFpaper} which applies the leximin principle to multi-resource settings  by equalizing tenants’ normalized shares of their bottleneck resource. However, DRF faces key limitations, particularly under congestion and in the presence of low-demand tenants, where it may fail to satisfy properties such as bottleneck saturation and share-incentive.

More precisely, in this paper we propose the \textit{ Dependency-aware Dominant Resource Fairness} (DDRF) algorithm, a centralized static allocation mechanism that extends DRF to handle congestion and generic inter-resource dependencies.
In DDRF, tenant satisfaction is determined per resource by introducing a distinct satisfaction variable for each tenant-resource pair, rather than a single satisfaction ratio per tenant.

We model dependencies through a family of dependency groups, where each group captures resources that are coupled for that tenant. Within each group, the effective allocation is determined by the congested bottleneck resource (i.e., the resource with the highest demand-to-capacity ratio, referred to as the dominant share). DDRF then equalizes the dominant shares of these bottleneck resources across tenants' dependency groups that share a common resource, while granting low-demand tenants their full demand, and redistributing any remaining surplus resource among the others according to MMF. This paper makes the following contributions:

\begin{itemize}
    \item We propose the {Dependency-aware Dominant Resource Fairness} (DDRF) algorithm, which introduces tenant-specific \emph{dependency groups} to capture generic couplings across resources. DDRF identifies low-demand tenants, reallocates surplus resources when they are present, and equalizes dominant shares of active congested resources across dependency groups that share a common resource.
    \item We introduce an \emph{effective satisfaction} metric to capture the used allocation's portion under dependencies. 
    \item We develop a practical DDRF solver that combines a convex heuristic with an evolutionary-optimization to handle convex and selected non-convex dependency cases.
    \item We devise a dependency-aware utilitarian baseline for fairness assessment.
    \item We evaluate DDRF on Amazon EC2 traces and a virtualized radio access network (vRAN) case study coordinating radio and computing resources across evolved Node Bs (eNBs). We consider demand profiles, dependency models, and congestion regimes, and benchmark against state-of-the-art baselines in terms of effective satisfaction, allocation waste and fairness.
\end{itemize}
The remainder of this paper is organized as follows. Section~\ref{backgroundandrelatedwork} reviews the background and related work. Section~\ref{SystemModel} introduces our system model. Section~\ref{ddrf} presents our proposed DDRF framework. Section~\ref{sec:simulation} reports the simulation settings and performance evaluation. Section~\ref{sec:results} provides experimental results, and Section~\ref{conclusion} concludes the paper.

\section{Background and related work}
\label{backgroundandrelatedwork}

Resource allocation has been extensively studied across diverse contexts. In its general form, when $N$ tenants submit demands for $M$ resources, the system is represented by a demand matrix $D \in \mathbb{R}_+^{N \times M}$, where $d_{ij} \in D$ denotes tenant $i$'s demand for resource $j$. The available capacities are given by a vector $C \in \mathbb{R}_+^M$, where $c_j \in C$ represents the total available capacity of resource $j$. The ordered pair $(D,C)$ defines a \emph{multi-resource allocation problem}.
The output of this problem is  given by $A \in \mathbb{R}_+^{N \times M}$ where $a_{ij} \in A$ is the amount of resource $j$ allocated to tenant $i$ such that $\sum_{i = 1}^N a_{ij} \leq c_j, \; \forall \; j$. 
The system is said to be \textit{congested} if there exists at least one resource $j$ such that $\sum_{i=1}^N d_{ij} > c_j$. Under congestion, not all tenants can be fully satisfied, making well-defined notions of fairness essential. It is worth mentioning that we consider the case where tenants cannot receive more than their demand, i.e., over-provisioning is not allowed in our system;  $0 \leq a_{ij} \leq d_{ij}$.

Most of the papers in the literature consider the case of single-resource allocation (i.e., $M=1$) or an oversimplified multiple-resource allocation schemes, in which strict linear proportionality between resources is assumed, resulting in allocations that preserve these fixed ratios. In other terms, under the linear-dependency scenario, the induced allocation is given by $a_{ij} = x_{i} \times d_{ij} $ where $x_i \in [0, 1]$ denotes the fraction of tenant $i$'s demand that is satisfied uniformly across all resources; we refer to this $x_i$ as the user satisfaction.

Several prominent fairness notions have exist for the single-resource setting, including Max-Min Fairness (MMF)~\cite{bertsekasbook,ogryczak2014fair,mmfnsdi}, Proportional Fairness (PF)~\cite{proportionalfairness,proportionalfairnessequation}, the more general $\alpha$-fairness that is also known as $(p,\alpha)$-Proportional Fairness~\cite{mo2002fair} which spans from utilitarian to PF and MMF, and some game-theoretic solutions such as Constrained Equal Awards, Constrained Equal Losses~\cite{moulin2002axiomatic} and Mood Value~\cite{moodvalue}.

A desirable property of any fairness rule is \emph{strategy-proofness}: tenants should not gain by misreporting their demand. Among the above, MMF is unique in being strategy-proof, since tenants with demands greater than or equal to $C/N$ cannot improve their allocation by cheating. Formally, a feasible single-resource allocation $(a_1, \dots, a_N)$ is \emph{max-min fair} if it is impossible to increase the allocation of any tenant $i$ without either violating feasibility or decreasing the allocation of some tenant $i'$ with $a_{i'} \leq a_i$. This notion is closely related to Rawls’ \emph{difference principle}~\cite{rawls1971theory}, which seeks to maximize the welfare of the least advantaged. MMF ensures that the weakest tenants (i.e., in terms of demands) are made as well off as possible. MMF allocation can be derived through a water-filling procedure; first, sort demands in non-decreasing order, $D_\delta = (d_{\delta(1)}, \dots, d_{\delta(N)})$, with $d_{\delta(1)} \leq \dots \leq d_{\delta(N)}$, then the allocation for tenant $\delta(i)$ is   recursively computed as $a_{\delta(i)} = \min \left(d_{\delta(i)}, \; \frac{C - \sum_{j=1}^{i-1} a_{\delta(j)}}{N - i + 1}\right).$ Intuitively, MMF follows an egalitarian principle by making the allocated amount equitable, and each tenant receives at least $C/N$, unless her demand is smaller, in which case she is allocated her full demand and the surplus is redistributed equally among the remaining tenants. 
Directly applying single-resource allocation rules to multi-resource settings often leads to inefficiency and allocation waste~\cite{MURANESPaper}, since tenants typically require multiple resources in combination rather than in isolation. This motivates the need for multi-resource allocation rules that extend single-resource fairness notions. 

Dominant Resource Fairness (DRF)~\cite{DRFpaper} generalizes MMF to multiple resource types, and was originally proposed for cloud systems. Recent works extends DRF to indivisible allocations~\cite{parkes2015beyond}, multiple servers~\cite{wang2014dominant}, and settings with elastic demands~\cite{chowdhury2016hug}. Other generalizations include axiomatic approaches based on Schur-concave functions~\cite{lan2010axiomatic,wang2014multi}, extensions of proportional fairness~\cite{BMFpaper}, and frameworks unifying multiple notions through ordered weighted averages (OWA)~\cite{MURANESPaper} that have re-adapts several fairness notions to network slicing environments. 

Despite their differences, all these approaches share a common limitation: they assume either linear coupling of resources  or Leontief-type utility preferences ($u_i  = \min_j \frac{a_{ij}}{d_{ij}}$) which treats resources as perfect complements where $u_i$ is the utility of user $i$ to maximize. This assumption severely restricts applicability in modern systems, where more general forms of inter-resource dependency arise.

Dominant Resource Fairness (DRF)~\cite{DRFpaper} generalizes MMF to multi-resource settings under linear proportional dependency assumption while retaining its desirable axiomatic properties. DRF operates by normalizing each tenant's $i$ demand through the demand-to-capacity ratio $s_{ij} = d_{ij}/c_j$ for every resource $j$. For each tenant $i$, the maximum of these demand-to-capacity ratios, $\mu_i = \max_j s_{ij}$, is called tenant $i$'s \emph{dominant share}, and the corresponding resource $b_i = \arg\max_j s_{ij}$ is the \emph{bottleneck resource}. Allocations are then determined by equalizing dominant shares across tenants, thereby extending MMF’s fairness guarantees to multi-resource scenarios.
Formally, DRF for multi-resource settings with linear dependency solves:
\begin{equation}
\label{eq:drf}
\tag{DRF}
\begin{aligned}
\max_{X} \;& \sum_{i = 1}^N x_i  \\ 
\text{s.t.}\;& D^\top X \le C, \\[0.5pt]
& \mu_i x_i = \mu_{j} x_{j}, \quad \forall\, i \ne j \in \{1, \cdots, N\}, \\[0.5pt]
& X \in [0, 1]^N.
\end{aligned}
\end{equation}
DRF maximizes each component in the vector $X$ but from $\mu_ix_i = \mu_j x_j = \lambda \iff x_i = \frac{\lambda}{\mu_i}$ that is increasing in $\lambda$ and no trade-off between the components we can say it  maximizes the sum of all users' satisfaction (Utilization) 
subject to capacity constraints, bounds on the satisfaction (non-negativity and no over-provisioning)
and a fairness constraint on equalizing dominant shares across tenants. DRF is designed to satisfy four key properties: \emph{Share Incentive}, guaranteeing that each tenant receives at least a $1/N$ fraction of her dominant resource if she demands more; \emph{Pareto Efficiency}, ensuring no tenant can be made better off without harming another; \emph{Strategy Proofness}, which prevents tenants from gaining by misreporting their demands; and \emph{Envy Freeness}, meaning every tenant prefers her own allocation to that of others. In essence, DRF reduces the multi-resource allocation problem to a single-resource one by focusing on the dominant share, and then scales allocations proportionally across other resources. Under congestion scenarios, the DRF scheme exhibits two key shortcomings:
\begin {itemize}
\item \textit{\textbf{Existence of weak tenants.}}
As DRF equalizes dominant shares, a \emph{weak tenant} with small demand can limit other users to its dominant share, preventing full resource utilization and leaving capacity idle.\\
\textit{\textbf{Example.}} 
Consider $2$ tenants and $2$ resources with {\footnotesize $D=\begin{bmatrix}9 & 9 \\ 14 & 25\end{bmatrix}$} and $C^\top=[20, 30]$, DRF produces the following allocation {\footnotesize $\begin{bmatrix}9 & 9 \\ 7.56 & 13.5 \end{bmatrix}$}, fully satisfying tenant~1 whereas tenant 2 reaches only $54\%$  satisfaction, even though the resources are not fully used.
However, 
a more efficient allocation than DRF's output is  {\footnotesize $\begin{bmatrix}9  & 9 \\ 11 & 19.6429\end{bmatrix}$} that fully utilizes resource~2 and raises tenant~2 satisfaction to $78.57\%$, strictly improving efficiency and preserving fairness. This can be seen geometrically in Figure~\ref{fig:drf_inefficiency}. The figure depicts the DRF feasible set \emph{with} the demand-bound constraint $x_i \le 1,\ \forall i\in\mathcal N$ (demand-satisfaction allocation) and \emph{without} it (task allocation). The grey and black points correspond to the DRF solutions without and with the constraint, respectively. The orange point indicates an alternative allocation that preserves DRF's properties yet Pareto-dominates the black point, by increasing the second user's allocation without reducing the first user's allocation.

\begin{figure}[t]
\centering
    \includegraphics[width=0.7\columnwidth, keepaspectratio]{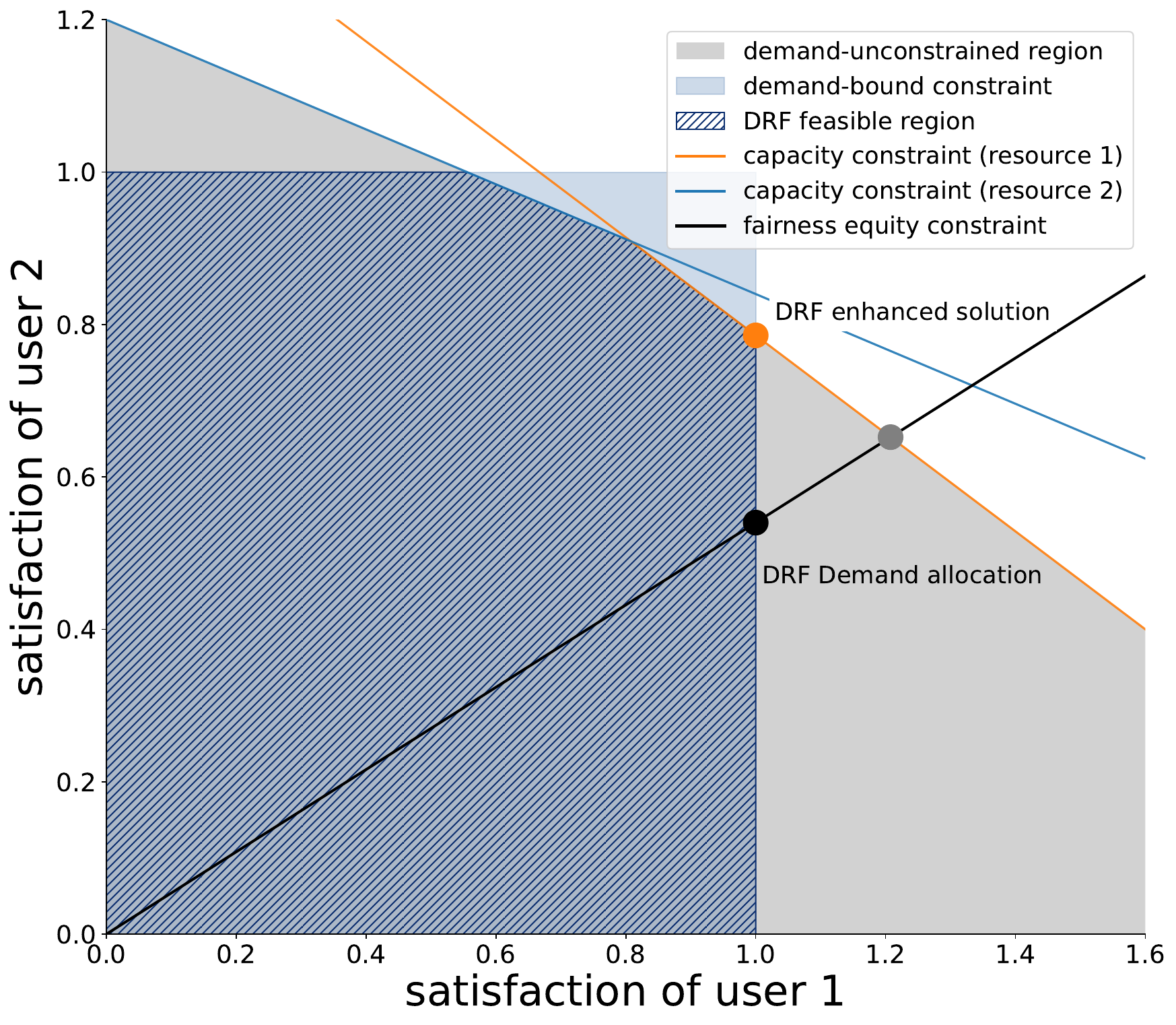}
    \vspace{-0.2cm}
    \caption{DRF inefficiency in demand satisfaction under congestion: allocation capped by demand. Fairness line prevents capacity saturation due to low-demand users; ignoring fairness yields efficient orange point.}
    \label{fig:drf_inefficiency}
\end{figure} 

\item \textit{\textbf{Dominant share on non-congested resources.}}
Finally, DRF may enforce fairness on a non-congested resource rather than the true bottleneck.

\textit{\textbf{Example.}} 
Consider 2 tenants and 2 resources with 
{\footnotesize$D=\begin{bmatrix}6 & 9 \\ 8 & 1\end{bmatrix}$ }
and 
$C^\top=[10,10]$. 
In this case, resource~1 is congested, while resource~2 still has spare capacity. 
DRF outputs the following allocation
{\footnotesize$\begin{bmatrix}4 & 6 \\ 6 & 0.75\end{bmatrix}$}. 
Although this satisfies DRF’s fairness condition, it does so without distinguishing between congested and non-congested resources. Alternatively, equalizing the dominant shares of the \textit{congested} resource (resource~1 in this example) gives 
{\footnotesize $\begin{bmatrix}5 & 7.5 \\ 5 & 0.625\end{bmatrix}$}. 
This adjustment aligns fairness with the actual bottleneck of the system, emphasizing that dominant shares should be defined with respect to congested resources rather than globally across all resources. 
\end{itemize}

To evaluate the efficiency of a resource allocation problem, the literature often defines efficiency as $\sum_i U_i(X)$, where $U$ is a strictly concave, increasing utility function that captures diminishing returns and fairness. A more general formulation is provided by the \textit{Ordered Weighted Averaging (OWA)} operator, adopted in the Multi-Resource Allocation for Network Slicing (MURANES) framework~\cite{MURANESPaper}. It is defined as  $\mathrm{OWA}(\mathbf{z}) = \sum_{i=1}^N w_i z_{\delta(i)},$ where $\delta$ orders the inputs in non-decreasing order and $w$ is a strictly decreasing weight vector. With this choice, OWA is Schur-concave, equitable, and satisfies the Pigou-Dalton transfer principle~\cite{dalton1920measurement}. By tuning the weights, OWA spans objectives from egalitarian $w=(1,0,\dots,0)$ to fully utilitarian $w=(1,1,\dots,1)$.  
Hence the relationship among fairness notions: MMF and PF are both \emph{egalitarian} in spirit, MMF enforces equity at the allocation level, while PF enforces equity at the satisfaction level. In contrast, utilitarian rules maximize total utility without fairness concerns.

In this paper, we address the aforementioned limitations by revisiting multi-resource allocation, which allows us to extend max–min fairness (MMF) to dependency-aware multi-resource settings while preserving its desirable properties.

\begin{table}[!t]
\centering
\renewcommand{\arraystretch}{1.43}
\caption{Summary of the notation used in the paper}
\vspace{-0.2 cm}
\label{notationstable}
\begin{tabular}{|m{0.005cm}|>{\centering\arraybackslash}m{0.7cm}|m{6.9cm}|}
\hline
& \hspace*{-0.15 cm}{\scriptsize \textit{\textbf{Symbol}}} & \multicolumn{1}{c|}{\textit{\textbf{Definition}}} \\
\hline

\multirow{10}{*}{\hspace*{-0.19cm}\rotatebox[origin=c]{-90}{\textit{System Model \ref{SystemModel} + \ref{DDRFAlgorithm}}}}
& {\footnotesize $\mathcal{N}$} & {\scriptsize $\{1,\cdots,N\}$ Set of users} \\
\cline{2-3}
& {\footnotesize $\mathcal{M}$} & {\scriptsize $\{1, \cdots, M\}$ Set of resources } \\
\cline{2-3}
& {\footnotesize $D$} & {\scriptsize $(d_{ij})_{1 \leq i \leq N, \; 1 \leq j \leq M}$, \; $d_{ij}:$ User's $i$ demand for resource $j$ } \\
\cline{2-3}
& {\footnotesize $C$ } & {\scriptsize $(c_j)_{1 \leq j \leq M}$, \; $c_j:$ The available capacity of resource $j$ } \\
\cline{2-3}
& {\footnotesize $X$} & {\scriptsize $(x_{ij})$,  $x_{ij} \in [0, 1]$: User's $i$ satisfaction for resource $j$ } \\
\cline{2-3}
&   {\footnotesize $K_i$} & {\scriptsize $\in \mathbb{N} \cup \{0\}$, Number of dependency functions of user $i$}\\
\cline{2-3}
& {\footnotesize $\mathcal{F}_i$} & {\scriptsize $ \Big\{
f_i^{(k)}\!\big(\{x_{ij}\}_{j \in S_i^{(k)}}\big) = 0 \; (\text{or } \leq 0) \; : \; k=1,\dots,K_i
\Big\},$  Set of dependency constraints coupling the resources of user $i$}  \\
\cline{2-3}
& {\footnotesize $S_{i}^{(k)}$ } & {\scriptsize $\subseteq \mathcal{M}$ Subset of coupled resources for user $i$ at constraint $k$  by $f_i^{(k)}$ } \\
\cline{2-3}
& {\footnotesize $A$}  & {\scriptsize $(a_{ij})$, $a_{ij} = d_{ij}x_{ij}$ The allocated amount of resource $j$ to user $i$ } \\
\cline{2-3}
& {\footnotesize $\lambda_j$ } & {\scriptsize $\in \mathbb{R}_+$ largest fully allocatable demand for resource $j$ under MMF} \\
\cline{2-3}
& {\footnotesize $y_{ij}$}  &  {\scriptsize $\mathbf{1}[d_{ij}>\lambda_j]$: indicator that user $i$'s demand for resource $j$ is active} \\
\cline{2-3}
&   $\mathcal{S}_i$     &  {\scriptsize Family of maximal dependency groups induced by $\{S_i^{(k)}\}_{k=1}^{K_i}$}  \\
\cline{2-3}
 & {\footnotesize $S$ } & { \scriptsize $(s_{ij}): s_{ij} = \frac{d_{ij}}{c_j}$ User $i$'s share of resource $j$} \\
 \cline{2-3}
 & {\footnotesize $\mu_i$}    & {\scriptsize $\max_j s_{ij}$ User $i$' dominant share} \\
\hline

\multirow{10}{*}{\hspace*{-0.15cm}\rotatebox[origin=c]{-90}{\textit{Theoretical Part \ref{TheoreticalPart}}}}
& {\footnotesize $b_i$}  & {\scriptsize $\in \argmax_{j \in \mathcal{M}} s_{ij}$ User $i$'s bottleneck resource} \\
\cline{2-3}
&  {\footnotesize $\mathcal{C}$}  & {\scriptsize $\{j \in \mathcal{M}: \sum_{i} d_{ij} > c_j\}$ The set of congested resources}    \\
\cline{2-3}
&   {\footnotesize $\mu_i^{\mathcal C}$} & {\scriptsize $\max_{j \in \mathcal{C}} s_{ij}$ User $i$'s dominant-share over congested resources} \\
\cline{2-3}
&   {\footnotesize $b_i^\mathcal{C}$}  &  {\scriptsize $\min \argmax_{j \in \mathcal{C}} s_{ij}$ User $i$' Congested bottleneck resource} \\ 
\cline{2-3}
& {\footnotesize $j'_i$} & {\scriptsize $\min \arg\max_{j\in\mathcal C:\,y_{ij}=1} s_{ij}$ User $i$' active congested bottleneck } \\
\cline{2-3}
& {\footnotesize $X^\star$}   &  {\scriptsize Optimal solution to \ref{DDRF} or \ref{utilitarian} } \\
\cline{2-3}
& {\footnotesize $\mathcal{A}_i(j')$} & {\scriptsize $\{\, k \in \{1,\dots,K_i\} : j' \in S_i^{(k)} \text{ and } f_i^{(k)}(X^\star)=0 \,\}$ active dependency-constraint indices of user $i$ containing $j'$ at $X^\star$} \\
\cline{2-3}
& {\footnotesize $\mathcal{B}\mathcal{C}$} & {\scriptsize $\{i \in \mathcal{N}: b_i \in \mathcal{C}\}$ Users with congested bottleneck resource} \\
\cline{2-3}
& {\footnotesize $\mathcal{B}\mathcal{N}\mathcal{C}$ } & {\scriptsize $\mathcal{N} \backslash \mathcal{B}\mathcal{C}$ Bottleneck-noncongested users} \\
\cline{2-3}
& {\footnotesize $\mathcal{W}$ } & {\scriptsize $\{i \in \mathcal N: \forall j \in \mathcal C y_{ij} = 0\}$ } Set of weak users \\
\cline{2-3}
& {\footnotesize $\mathcal{A}$ } & {\scriptsize $\mathcal N \backslash \mathcal W $ Set of active users } \\
\cline{2-3}
& \hspace*{-0.2 cm}{\footnotesize $M_1(\alpha; z)$ } & {\scriptsize $ M_1(\alpha; z) = \frac{\sum_i \alpha_i z_i}{\sum_i \alpha_i}$ with $\alpha, z \in \mathbb{R}^N$ The weighted mean } \\
\hline

\multirow{4}{*}{\hspace*{-0.15cm}\rotatebox[origin=c]{-90}{\textit{Simulation \ref{sec:simulation}}}}
& {\footnotesize $CP$ } & {\scriptsize $[0, 1]^M$ Congestion profile, $CP_j$ available proportion of $\sum_i d_{ij}$ } \\
\cline{2-3}
& \hspace*{-0.2 cm}{\footnotesize $\mathcal{E}(X, \mathcal{F})$ } & {\scriptsize $\left\{ e \in [0,1]^{N \times M} \;\middle|\; 0 \leq e_{ij} \leq X_{ij}, \; \forall i,j, \;\; e \in \mathcal{F} \right\}$ Effective satisfaction region} \\
\cline{2-3}
    & \hspace*{-0.2 cm}{\footnotesize  $X^{\text{effective}}$} & {\scriptsize $\argmax_{e \in \mathcal{E}(X, \mathcal{F})} \;\; \sum_{i \in \mathcal{N}} \sum_{j \in \mathcal{M}} e_{ij}$ Effective satisfaction} \\
\hline

\end{tabular}
\end{table}

\section{System model and problem formulation}
\label{SystemModel}
We consider a network slicing problem where each tenant manages an isolated slice made up of virtualized resources that are provided by one or more infrastructure providers (InPs). Tenants may correspond to MVNOs  (Mobile Virtualized Network Operator) who does not have their own network infrastructure, service providers, vertical industries, or end users with specific QoS requirements. InPs manage the underlying physical infrastructure and allocate resources across slices. Network slicing thus enables customized, isolated virtual networks while ensuring efficient infrastructure sharing and meeting service-level guarantees. For convenience, the main notation used in the paper is summarized in Table~\ref{notationstable}. 

Our system model is shown in Figure~\ref{fig:network_slicing}. Let $\mathcal{N} = \{1, \dots, N\}$, $\mathcal{M} = \{1, \dots, M\}$, and $\mathcal{P} = \{1, \dots, P\}$ be tenants, resources, and infrastructure providers (InPs). 

\begin{figure}
\centering
    \includegraphics[width=\columnwidth]{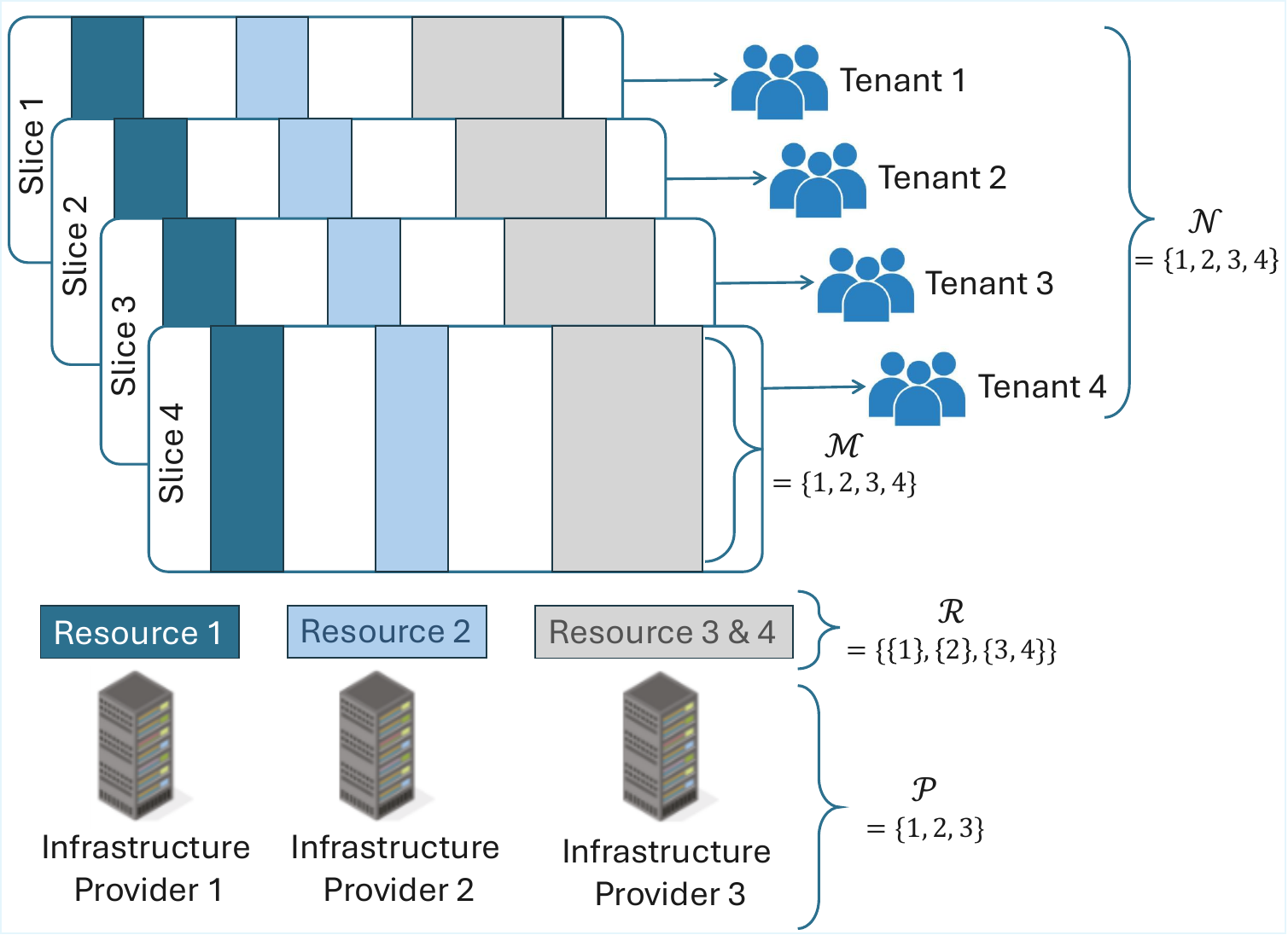}
    \vspace{-0.7cm}
    \caption{A system model example with  4 tenants, 3 InPs, 4 resources.}
    \label{fig:network_slicing}
\end{figure} 

The tuple $(D, C, \mathcal{F})$ defines the dependency-aware multi-resource allocation problem, with demand matrix $D \in \mathbb{R}_+^{N \times M}$, where $d_{ij}$ denotes tenant $i$'s demand for resource $j$. The available capacities are given by a vector $C \in \mathbb{R}_+^M$, where $c_j$ is the total capacity of resource $j$. We adopt Per-Resource Satisfaction\footnote{Unless otherwise noted, all subsequent definitions adopt this formulation.},   
with $X \in \mathbb{R}^{N \times M}$, where each $x_{ij} \in [0,1]$ denotes the fraction of demand $d_{ij}$ satisfied for tenant $i$ on resource $j$. 
This more general formulation captures arbitrary, potentially non-linear couplings between resources. It enables us to define satisfaction per resource, which is essential for modeling dependencies beyond simple proportionality and is the foundation of our approach. Each tenant $i$ specifies a collection of $K_i$ functional constraints:
$$
\mathcal{F}_i \triangleq \Big\{
f_i^{(k)}\!\big(\{x_{ij}\}_{j \in S_i^{(k)}}\big) = 0 \; (\text{or } \leq 0) \; : \; k=1,\dots,K_i
\Big\},
$$
where each constraint captures a dependency among a subset of resources 
$S_i^{(k)} \subseteq \mathcal{M}$.
An allocation is \emph{valid} for tenant~$i$ only if all constraints in $\mathcal{F}_i$ are satisfied.
Let $\mathcal{F} = \bigcup_{i=1}^N \mathcal{F}_i$ denote the global set of all dependency constraints. We assume that the fully satisfied demands is a feasible candidate solution (i.e., $x_{ij}=1$ for all $j$) is feasible for all $i, k$, so $f_i^{(k)}\!\big(\{\mathbf{1}\}_{j \in S_i^{(k)}}\big)\, = 0$ $(\text{or } \leq 0) \; \forall \; i, k$, where  $\mathbf{1}$ is the vector of ones, this model assumption ensures that demands are consistent with the underlying dependencies between resources.\\
\textit{\textbf{Example.}} 
Consider tenant~$i$ with demand vector $(2,4,3,6)$ across four resources, subject to the following dependencies:\\
$
f_i^{(1)}(x_{i1}, x_{i2}) = (2x_{i1})^2 - 4x_{i2} = 0 \quad \text{(quadratic)},\\
f_i^{(2)}(x_{i1}, x_{i3}, x_{i4}) = 2x_{i1}\cdot 3x_{i3} - 6x_{i4} = 0 \quad \text{(multiplicative)}.$
These constraints capture nonlinear couplings among resources that cannot be represented by fixed linear ratios. We note that $f^{(1)}_i(1,1) = 0$ and $f_i^{(2)}(1, 1, 1) = 0$ meaning that tenants are rational, so their declared demands are consistent with the specified dependencies; hence, the fully satisfied demand vector is always feasible. Importantly, \emph{monotonicity is not required}: our model admits both increasing and decreasing relationships. The latter capture trade-offs where allocating more of one resource reduces the need for another, while our optimization framework ensures efficiency by maximizing overall satisfaction across all resources. The output of this problem is  given by $A \in \mathbb{R}_+^{N \times M}$ where $a_{ij} = x_{ij} \cdot d_{ij}, \; \forall i \in \mathcal{N}, \; j \in \mathcal{M}$ is the amount of resource $j$ allocated to tenant $i$ such that $\sum_{i = 1}^N a_{ij} \leq c_j, \; \forall \; j  \in \mathcal{M}$.

\begin{algorithm}[tb]
\caption{Full Allocation Cutoff Computation}
\label{alg:thresholds}
\begin{algorithmic}[1]
\Require Demand $D \in \mathbb{R}^{N \times M}_+$, Capacity $C \in \mathbb{R}^M_+$
\Ensure Threshold vector $\boldsymbol{\lambda} \in \mathbb{R}^M$
\For{$j = 1$ to $M$}
    \State $\mathbf{d}_j \gets [d_{1j}, \dots, d_{Nj}]$ \Comment{demands on resource $j$}
    \State Sort $\mathbf{d}_j$: $d_{(0)j} \gets 0 \le d_{(1)j} \le \dots \le d_{(N)j}$
    \State $S \gets 0$; \quad $\textit{found} \gets \textbf{false}$ \Comment{prefix sum; termination}
    \For{$k = 0$ to $N-1$}
        \State $\tilde{\lambda} \gets \dfrac{c_j - S}{N - k}$ \vspace{0.1 cm}
        \If{$d_{(k)j} \le \tilde{\lambda} \le d_{(k+1)j}$}
            \State $\lambda_j \gets \tilde{\lambda}$; \quad $\textit{found} \gets \textbf{true}$; \quad \textbf{break}
        \EndIf
        \State $S \gets S + d_{(k+1)j}$
    \EndFor
    \If{\textit{found} $=\textbf{false}$}
        \State $\lambda_j \gets d_{(N)j}$ \Comment{all demands can be fully satisfied}
    \EndIf
\EndFor
\State \Return $\boldsymbol{\lambda} = [\lambda_1, \dots, \lambda_M]$
\end{algorithmic}
\end{algorithm}

We consider a centralized setting in which a central entity receives $(D, C, \mathcal{F})$ and computes the allocation. For simplicity, we adopt the standard assumption that resources are divisible\footnote{We assume each resource is exclusively managed by a single provider to avoid inter-provider contention. Modeling shared-resource scenarios is out of the scope of this paper.}.

Our approach aims at maximizing system efficiency interpreted as the linear sum of user utilities. In our network slicing model, utility is derived directly from user satisfaction on the allocated resources. The utility of user $i$ is defined as $U_i = \sum_{j \in \mathcal{M}} x_{ij}$, with $0 \leq U_i \leq M$ for all $i \in \mathcal{N}$. The total system efficiency, or \textit{user-centric efficiency}, is given by $ \sum_{i \in \mathcal{N}} U_i = \sum_{i \in \mathcal{N}} \sum_{j \in \mathcal{M}} x_{ij}$.

To address one of the main limitations of the DRF solution, namely the presence of weak tenants that lead to inefficiencies, we first formalize the notion of \emph{weak tenants}.

To do so, we compute a threshold $\lambda_j$ for each resource $j \in \mathcal{M}$. Intuitively, $\lambda_j$ represents the largest per-resource demand that can be fully satisfied under the Max-Min Fairness (MMF) principle. Formally, for each resource $j$ with demand vector $[d_{1j}, \dots, d_{Nj}]$ and capacity $c_j$, the threshold $\lambda_j$ is defined such that all tenants with $d_{ij} \leq \lambda_j$ receive their full demand.
\begin{definition}[Weak Tenants]
A tenant $i$ is classified as \emph{weak} if $d_{ij} \leq \lambda_j$ for all $j \in \mathcal{M}$. 
\end{definition}

Weak tenants are automatically satisfied, while only the remaining \emph{active tenants} participate in the fairness constraint. The thresholds $\lambda_j$ are computed using Algorithm~\ref{alg:thresholds}, which follows the MMF water-filling procedure with time complexity $O(MN \log N)$.

\section{Dependency-aware Dominant Resource Fairness (DDRF)}
\label{ddrf}

Motivated by the aforementioned objectives: (i) resolving DRF’s inefficiency and (ii) accommodating generic inter-resource dependencies, we propose a resource-allocation framework for congested systems with arbitrary inter-resource dependencies that overcomes the key limitations of DRF.
\vspace{-0.3 cm}

\subsection{Modeling Coupled Resources via Dependency Groups}
\label{DDRFAlgorithm}
 
 As a first step, we introduce the concept of \emph{dependency groups}, an abstraction used purely for algorithmic purposes. Although in practical systems such as cloud and network slicing it is unrealistic for group of resources to be coupled independently of other groups, this abstraction enables a unified and tractable formulation of our algorithm.
 
 \begin{definition}[User Dependency Family]
 Given a user $i$, her dependency constraints $\mathcal{F}_i$ and $S_i^{(k)}$ the set of coupled resources of each constraint $k$.  The Dependency Family $\mathcal{S}_i$ is formed by merging all overlapping $S_i^{(k)}$ over constraints $k$, each resulting set is a \emph{dependency group}. Disjoint interdependent sets form separate groups, and unconstrained resources appear as singletons.
 \end{definition}

These families, derived from $\mathcal{F}$, enable \textsc{DDRF} to enforce fairness within each coupled set of resources, even when coupling is only partial across $\mathcal{M}$ .

\emph{\textbf{Example.}} Suppose tenant $i$ requests $7$ resources with the following dependency constraints:\\
$
2(d_{i2}x_{i2}) + 5 - d_{i4}x_{i4} = 0, \qquad
8(d_{i1}x_{i1})^3 - d_{i6}x_{i6} = 0, \\
3(d_{i1}x_{i1})^2 + (d_{i5}x_{i5})(d_{i3}x_{i3}) = 20.\\
$
Then, $\mathcal{S}_i = \{\{2,4\}, \{1,3,5,6\}, \{7\}\}$.

\textit{\textbf{Special cases.}} \textit{Linear dependency:} if $x_{ij}=x_{ik}$ for all $j,k \in \mathcal{M}$, then $\mathcal{S}_i=\{\mathcal{M}\}$ i.e., all resources are coupled in a single dependency group. \textit{No dependency:} if all resources are independent, then $\mathcal{S}_i=\{\{j\}: j \in \mathcal{M}\}$.

\begin{algorithm}[t]
\caption{Fairness Parameters of DDRF}
\label{alg:DDRF}
\begin{algorithmic}[1]
\Require Demand $D \in \mathbb{R}^{N \times M}_+$, Capacity $C \in \mathbb{R}^M_+$, Dependency $\mathcal{F}$
\Ensure Fairness parameters $(\hat{x}, \hat{\mu}, \hat{y})$
\State Compute normalized shares $s_{ij} \gets d_{ij}/c_j$
\State Compute thresholds $\lambda_j$ for each resource $j$ (Alg.~\ref{alg:thresholds})
\State Activity matrix $y_{ij} \gets \mathbf{1}[d_{ij} > \lambda_j]$
\State For each tenant $i$, obtain dependency family $\mathcal{S}_i$
\For{$i=1$ to $N$}
  \ForAll{$S \in \mathcal{S}_i$}
    \State $\mathcal{J} \gets \{ j \in S : y_{ij}=1 \}$ \Comment{active indices in $S$}
    \If{$\mathcal{J}=\emptyset$}
        \State $\mathcal{J} \gets S$ \Comment{any index in $S$ can be chosen}
    \EndIf
    \State $j^\star \gets \min \arg\max_{j \in \mathcal{J}} s_{ij}$
    \ForAll{$j \in S$}
        \State $\hat{y}_{ij} \gets y_{i j^\star}$ \Comment{inherit activity from $j^\star$}
        \State $\hat{\mu}_{ij} \gets s_{i j^\star}$ \Comment{inherit dominant share}
        \State $\hat{x}_{ij} \gets x_{i j^\star}$ \Comment{link satisfactions}
    \EndFor
  \EndFor
\EndFor
\State \Return $\hat{x}, \hat{\mu}, \hat{y}$
\end{algorithmic}
\end{algorithm}

The \textsc{DDRF} allocation problem can be formulated as the following optimization program:
\vspace{-0.1 cm}
\begin{equation}
\begin{aligned}
\max \; & \sum_{i=1}^N \sum_{j=1}^M x_{ij}  \\
\text{s.t.} \quad 
    & \sum_{i=1}^N d_{ij} x_{ij} \le c_j, \forall j \in \mathcal{M}, \hspace{1.96 cm} (1)
\\
& X \in \mathcal{F}, \hspace{4.58 cm} (2)\\
& \forall\, i\neq k\in\mathcal N,\; j\in\mathcal M: \hspace{2.4 cm} (3) \quad \\
& \quad \Delta_{ikj} = \max\!\big(\hat{\mu}_{ij},\hat{\mu}_{kj}\big)\big(1-\hat{y}_{ij}\hat{y}_{kj}\big), 
\\
& \quad \hat{\mu}_{ij}\hat{x}_{ij} - \hat{\mu}_{kj}\hat{x}_{kj} \le\Delta_{ikj},
\\
& \quad \hat{\mu}_{ij}\hat{x}_{ij} - \hat{\mu}_{kj}\hat{x}_{kj} \ge - \Delta_{ikj}, 
\\
& \hat{x}_{ij} \ge 1-\hat{y}_{ij}, \quad \forall\, i\in\mathcal{N},\; j\in\mathcal{M} \hspace{0.95 cm} (4) \\
& 0 \le x_{ij} \le 1, \forall i \in \mathcal{N},\; j \in \mathcal{M}. \hspace{1.45 cm} (5)
\end{aligned}
\tag{DDRF}
\label{DDRF}
\end{equation}

Here $(\hat{x}, \hat{\mu}, \hat{y})$ denote the satisfaction, dominant-share, and activity values inherited from each dependency group's active dominant resource (see Algorithm~\ref{alg:DDRF}).

The optimization program~\eqref{DDRF} maximizes tenants' satisfaction $\sum_{i\in\mathcal{N}}\sum_{j\in\mathcal{M}} x_{ij}$ subject to five classes of constraints.
(1) the capacity constraints ensure that the allocated amount of each resource does not exceed its available capacity.
(2) the dependency constraints $X \in \mathcal{F}$ enforce tenant-specific couplings across resources, so allocations remain consistent with the declared inter-resource relations.
(3) the fairness constraints implement the DDRF generalization of DRF. For each tenant, Algorithm~\ref{alg:DDRF} maps dependency group to its active dominant (bottleneck) resource and represents the group satisfaction accordingly. DDRF then enforces fairness equity $\hat{\mu}_{ij}\hat{x}_{ij}=\hat{\mu}_{kj}\hat{x}_{kj}$ when the corresponding groups are active (i.e., both tenants have at least one active resource in the group); when a group has no active resource, the constraint is relaxed and (4) full satisfaction is activated within that group.
(5) the bounds $0 \le x_{ij} \le 1$, ensure non-negative satisfaction and prevent over-provisioning beyond demand.

Algorithm~\ref{alg:DDRF} constructs the fairness parameters by defining the activity matrix
$y_{ij}=\mathbf{1}[d_{ij}>\lambda_j]$, where $y_{ij}=1$ indicates that tenant~$i$ is \emph{active} on resource~$j$ (its demand exceeds the MMF cutoff $\lambda_j$), and $y_{ij}=0$ otherwise.
Together with the normalized demand shares $s_{ij}=d_{ij}/c_j$, this allows \textsc{DDRF} to select, within each dependency group $S\in\mathcal{S}_i$, a \emph{representative} resource index $j^\star$ that drives the group’s fairness treatment.
Specifically, if the group contains at least one active resource, $j^\star$ is chosen as the active resource with the largest $s_{ij}$; ties are broken by taking the smallest index.
If no resource in the group is active, then the group is deemed \emph{inactive} (and the tenant is weak on that group), so $\hat{y}_{ij}=0$ for all $j\in S$ and no fairness equalization is enforced for that group; moreover, its full satisfaction is guaranteed.
The representative index $j^\star$ induces the inherited quantities used in~\eqref{DDRF}: for all $j\in S$, we set $\hat{y}_{ij}=y_{i j^\star}$, $\hat{\mu}_{ij}=s_{i j^\star}$, and $\hat{x}_{ij}=x_{i j^\star}$, so that the group is effectively governed by a single satisfaction variable.

\vspace{-0.3 cm}

\subsection{DDRF Properties}
\label{TheoreticalPart}
We now analyze the properties and features of \textsc{DDRF}.

\subsubsection{Pareto Efficiency}
A key requirement for any resource allocation algorithm is \emph{Pareto efficiency}: no user can be made better off without making another worse off. In multi-resource settings under congestion, a necessary and sufficient condition for Pareto efficiency is saturation i.e. to fully allocate as much as possible of the congested resources (otherwise, some tenant's satisfaction can be increased without harming others).

We assume $\mathcal{F}$ satisfies two mild properties: local regularity and the existence of an improving feasible direction. 
We first present the main theorem and a supporting lemma, then state the intuition behind these assumptions.

\begin{theorem}[Pareto-Efficiency of \textsc{DDRF}]\label{theorem1}
Under standard regularity assumptions on $\mathcal{F}$, every optimal solution of~\eqref{DDRF} for $(D, C, \mathcal{F})$ saturates at least one congested resource. Hence \textsc{DDRF} is Pareto efficient over its feasible set.
\end{theorem}

\begin{definition}[Dependency-Aware Utilitarian Framework]
By dropping the fairness constraint in~\eqref{DDRF}, we obtain the following dependency-aware utility maximization problem:
\begin{equation}
\begin{aligned}
\max \; & \sum_{i=1}^N \sum_{j=1}^M x_{ij}  \\
\text{s.t.} \quad 
& \sum_{i=1}^N d_{ij} x_{ij} \le c_j, && \forall j \in \mathcal{M}, \\
& X \in \mathcal{F}, \\
& 0 \le x_{ij} \le 1, && \forall i \in \mathcal{N},\; j \in \mathcal{M}.
\end{aligned}
\tag{D-Util}
\label{utilitarian}
\end{equation}
\end{definition}

\begin{lemma}[Pareto-Efficiency of~\ref{utilitarian}]\label{lemma1}
Under the same assumptions on $\mathcal{F}$ as in Theorem~\ref{theorem1}, every optimal solution of~\ref{utilitarian} saturates at least one congested resource, and is Pareto efficient.
\end{lemma}

\textit{Local regularity and local improving feasible direction.}
We require that certain functions in \(\mathcal{F}\) satisfy two mild assumptions. 
The first is \emph{local regularity}, meaning that the function and its first-order derivatives are continuous in a neighborhood of the optimal solution \(X^\star\) of~\eqref{DDRF} or~\eqref{utilitarian}; this allows the corresponding partial derivatives to be well defined and used in the analysis. 
The second is a \emph{local improving feasible direction} condition, which ensures that, under the dependency constraints, an increase in one coordinate can be accompanied by compensating variations in the coupled coordinates so that the overall feasible variation remains locally improving. 
The precise identification of the functions in \(\mathcal{F}\) to which these assumptions apply, together with the exact formal statements, is given in Appendix~\ref{appendixA}.

The detailed proofs of Lemma~\ref{lemma1} and Theorem~\ref{theorem1} are deferred to appendix~\ref{appendixB} and \ref{appendixC}.
Intuitively, if congested resources were left unsaturated, then under the assumptions on \(\mathcal{F}\) one could construct a small feasible perturbation, consistent with the dependency constraints, that strictly increases \(\sum_{i,j} x_{ij}\) while preserving feasibility, contradicting optimality.
Hence, both the dependency-aware utilitarian problem and \textsc{DDRF} admit only optimal solutions that saturate at least one congested resource.

\subsubsection{Sufficient Full Utilization}
When low-demand (weak) tenants are present, classical DRF may leave congested resources unsaturated. Under \emph{linear} dependencies, where each tenant receives the same satisfaction level across all resources (i.e., \(x_{ij}=x_i\) for all \(j\)), and with \(\mathcal C\) the set of congested resources as in table~\ref{notationstable}, \textsc{DDRF} reduces to the following scalar formulation \((X\in\mathbb R^N)\):
\[
\begin{aligned}
\max \quad & \sum_{i \in \mathcal{N}} x_i \\
\text{s.t.} \quad & D^\top X \le C, \\
&\mu_i^\mathcal{C} x_i -\mu_k^\mathcal{C} x_k \leq \max(\mu_i^\mathcal{C}, \mu_k^\mathcal{C})(1 -  y_{i b_{i}^\mathcal{C}}y_{k b_{k}^\mathcal{C}}) ,  \; \forall i \neq k\\
&\mu_i^\mathcal{C} x_i -\mu_k^\mathcal{C} x_k \ge -\max(\mu_i^\mathcal{C}, \mu_k^\mathcal{C})(1 -  y_{i b_{i}^\mathcal{C}}y_{k tb_{k}^\mathcal{C}}) ,  \; \forall i \neq k\\
& x_i \ge 1-y_{ib_{i}^\mathcal{C}}, \quad \forall\, i\in\mathcal{N} \\
& 0 \le x_i \le 1, \quad \forall i \in \mathcal{N},
\end{aligned}
\]

where $\mu_i^\mathcal{C}$ the share of the congested bottleneck and $b_i^\mathcal{C}$ the congested bottleneck resource are defined in table~\ref{notationstable}. Only tenants with $y_{i b_i^\mathcal{C}}=1$ (active on their congested bottleneck) are equalized; weak tenants are fully satisfied.

\subsubsection{Congestion-Aware Bottleneck Selection}
In \textsc{DDRF}, fairness is enforced only along \emph{limiting} (congested) resources. 
This anchors each dependency group to its most limiting \emph{active} resource and avoids driving fairness decisions by non-limiting coordinates, which can improve efficiency compared to DRF.

We introduce two classes of users based on their bottleneck resource The first one is users whose bottleneck resource is congested $\mathcal{B}\mathcal{C}$ and the second is the bottleneck-noncongested users $\mathcal{B}\mathcal{N}\mathcal{C}$ defined in table~\ref{notationstable}. Under linear dependencies and $\mathcal C\neq\emptyset$, the following scenarios occur:
\vspace{-0.4 cm}

{\footnotesize
\[
\renewcommand{\arraystretch}{1.18}
\setlength{\tabcolsep}{2.5pt}
\begin{array}{c|c|c}
\textbf{Case} & \textbf{Active} & \textbf{Weak} \\
\hline
\begin{array}{@{}c@{}}
(\mathcal{B}\mathcal{N}\mathcal{C}\\
=\emptyset)
\end{array}
&
\begin{array}{@{}l@{}}
\forall\, i\in\mathcal{B}\mathcal{C}:\\
y_{i b_i}=1
\end{array}
&
\begin{array}{@{}l@{}}
\exists\, i\in\mathcal{B}\mathcal{C}:\\
y_{i b_i}=0
\end{array}
\\
\hline
\begin{array}{@{}c@{}}
(\mathcal{B}\mathcal{N}\mathcal{C}\\
\neq\emptyset)
\end{array}
&
\begin{array}{@{}l@{}}
\forall\, i\in\mathcal{B}\mathcal{N}\mathcal{C}\ \text{and}\ i\in\mathcal{B}\mathcal{C}:\\
\exists\, j\in\mathcal C:\ y_{ij}=1\\
\text{(i)}\ \dfrac{M_1(\alpha^\mathcal{C}; d_{\cdot j^\star})}{c_{j^\star}}\\
\hspace{0.5 cm} \leq
\dfrac{M_1(\alpha; d_{\cdot j'})}{c_{j'}}\\
\text{(ii)}\ \dfrac{M_1(\alpha^\mathcal{C}; d_{\cdot j^\star})}{c_{j^\star}}\\
\hspace{0.5 cm} >
\dfrac{M_1(\alpha; d_{\cdot j'})}{c_{j'}}
\end{array}
&
\begin{array}{@{}l@{}}
\exists\, i\in\mathcal{B}\mathcal{N}\mathcal{C},\ \text{or}\ i\in\mathcal{B}\mathcal{C}:\\
\forall\, j\in\mathcal C:\ y_{ij}=0\\
\text{(i)}\ 
|\mathcal W|+\dfrac{\tilde c_{\tilde j^\star}}
{M_1(\alpha_{\mathcal A}^{\mathcal C};d_{\mathcal A,\cdot \tilde j^\star})}\\
\hspace{0.5 cm} \ge
\dfrac{c_{j'}}{M_1(\alpha;d_{\cdot j'})}\\
\text{(ii)}\ 
|\mathcal W|+\dfrac{\tilde c_{\tilde j^\star}}
{M_1(\alpha_{\mathcal A}^{\mathcal C};d_{\mathcal A,\cdot \tilde j^\star})}
\\
\hspace{0.5 cm}<
\dfrac{c_{j'}}{M_1(\alpha;d_{\cdot j'})}
\end{array}
\end{array}
\]
}

\vspace{-0.2 cm}

Where $\mathcal W$ the set of weak users, $\mathcal A$ the set of active users, $\mu_i$ the dominant share, $\mu_i^{\mathcal C}$ the congested dominant share and $M_1(\beta;z)$ is the weighted mean. They are defined in table~\ref{notationstable}.

Set $\mathcal M_0=\mathcal M\cup\{0\}$ and $d_{i0}=1$ for all $i\in\mathcal N$, and define $d_{\mathcal A,\cdot j}=(d_{ij})_{i\in\mathcal A}$ for each $j\in\mathcal M_0$. Associate the weights $\alpha_i=1/\mu_i$, $\alpha_i^{\mathcal C}=1/\mu_i^{\mathcal C}$, and $\alpha_{\mathcal A}^{\mathcal C}=(\alpha_i^{\mathcal C})_{i\in\mathcal A}$. Next, define $c_0=\big(\min_{i\in\mathcal N}\mu_i\big)\sum_{i\in\mathcal N}\alpha_i$, $c_0^{\mathcal C}=\big(\min_{i\in\mathcal N}\mu_i^{\mathcal C}\big)\sum_{i\in\mathcal N}\alpha_i^{\mathcal C}$, and $c_j^{\mathcal C}=c_j$ for all $j\in\mathcal M$. In the weak case, define the residual capacities $\tilde c_j=c_j-\sum_{i\in\mathcal W}d_{ij}$ for all $j\in\mathcal M$, and $\tilde c_0=\big(\min_{i\in\mathcal A}\mu_i^{\mathcal C}\big)\sum_{i\in\mathcal A}\alpha_i^{\mathcal C}$. Finally, define $j^\star\in\arg\min_{j\in\mathcal M_0}\frac{M_1(\alpha^{\mathcal C};d_{\cdot j})}{c_j^{\mathcal C}}$, $j'\in\arg\min_{j\in\mathcal M_0}\frac{M_1(\alpha;d_{\cdot j})}{c_j}$, and $\tilde j^\star\in\arg\min_{j\in\mathcal M_0}\frac{M_1(\alpha_{\mathcal A}^{\mathcal C};d_{\mathcal A,\cdot j})}{\tilde c_j}$.

\begin{theorem}\label{theorem2}
Under linear dependencies, let $x^{\mathrm{DDRF}}$ and $x^{\mathrm{DRF}}$ denote the solutions of DDRF and DRF, respectively. 
Then $\sum_{i\in\mathcal N} x_i^{\mathrm{DDRF}} \ge \sum_{i\in\mathcal N} x_i^{\mathrm{DRF}}$ in all scenarios except in the two cases labeled (ii) under $\mathcal{B}\mathcal{N}\mathcal{C}\neq\emptyset$ (Weak and Active).
\end{theorem}

\textbf{Example.} Consider $2$ users and $2$ resources, with {\footnotesize$D=\begin{bmatrix}4 & 8\\ 7 & 1\end{bmatrix}$} and $C^\top=[10,10]$, resource~1 is congested while resource~2 is not. 
In this instance, $\mu=(0.8,0.7)$ and $\mu^{\mathcal C}=(0.4,0.7)$, hence $\alpha=(1.25,1.4286)$ and $\alpha^{\mathcal C}=(2.5,1.4286)$.
With $\mathcal M_0={0,1,2}$, the minimizers are $j^\star=1$ and $j'=1$, and we obtain $\frac{M_1(\alpha^{\mathcal C};d_{\cdot 1})}{c_1}\approx 0.6585 < \frac{M_1(\alpha;d_{\cdot 1})}{c_1}\approx 0.7222$.
Thus condition (active, $\mathcal B \mathcal N \mathcal C \neq \emptyset$) (i) holds, so $\sum_i x_i^{\mathrm{DDRF}}>\sum_i x_i^{\mathrm{DRF}}$ (DDRF is more efficient here). The weighted congested pressure at DDRF's bottleneck resource (with weights $\frac{1}{\mu_i^C}$) is not larger than the pressure of DRF (with weights $\frac{1}{\mu_i}$). The proof of Theorem~\ref{theorem2} is deferred to appendix~\ref{AppendixD}. 

\subsection{Numerical Example}
\label{sec:anexample}

We consider $3$ tenants with $3$ resources that are the number of physical resource blocks $N_{\text{PRB}}$, the CPU clock frequency $f$ and the fronthaul bandwidth $B^{\text{FH}}$. Demands are as follows: 

     slice 1: $(60, 2.1, 1209.6)$, 
     slice 2: $(45, 2.22, 453.6)$ and
     slice 3: $(30, 1.25, 151.2)$
     
where each tuple corresponds to $(N_{\text{PRB}}, f_{\text{[GHz]}}, B_{\text{[Mbps]}}^{\text{FH}})$. The aggregate demand is $(135, 5.371, 1814.4)$ exceeds the available capacities $(106, 3.5, 1000)$, so all resources are congested. We have a linear proportional dependency between $N_{\text{PRB}}$ and $B_{\text{[Mbps]}}^{\text{FH}}$ $(x_{i1} = x_{i3})$ for each user $i \in \{1, 2, 3\}$ and a quadratic dependency between  $N_{\text{PRB}}$ and $f_{\text{[GHz]}}$ $(\alpha_i x_{i1} \leq x_{i2}^2, \; \alpha = (0.9992, 0.9921, 0.9733))$ for each user as well. This choice is inspired from realistic dependencies reported in O-RAN literature~\cite{khatibi2018modelling}, \cite{hojeij2025flexible}.
The dependency-aware multi resource allocation problem becomes $(D, C, \mathcal{F})$ with  $D = ${\footnotesize$\begin{bmatrix} 60 & 2.054 & 1209.6 \\ 45 & 2.22 & 453.6 \\ 30 & 1.097 & 151.2
\end{bmatrix}$}
where the columns correspond to the resources and the rows to the tenants. The available resources are given by $C^\top =${ \footnotesize $ \begin{bmatrix} 106 & 3.5 & 1000 \end{bmatrix}$}. 

To obtain \textsc{DDRF} solution we compute the active demands matrix $y =$ {\footnotesize $\begin{bmatrix} 1 & 1 & 1 \\ 1 & 1 & 1 \\ 0 & 0  & 0 \end{bmatrix}$} and the share matrix  $s =  $ {\footnotesize$ \begin{bmatrix} 0.566 & 0.5869 & 1.2096 \\ 0.4245 & 0.6343 & 0.4536 \\ 0.283 & 0.3134 & 0.1512 \end{bmatrix}$}. 
Clearly, user 1 has bottleneck on resource 3 with $\mu_1 = 1.2096$, user 2 has bottleneck on resource 2 with $\mu_2 = 0.6343$ and user 3 is labelled as weak.
DDRF allocation is obtained by solving the following program:
\begin{IEEEeqnarray}{rCl}
&&\max_{X}  \displaystyle \sum_{i=1}^{3}\sum_{j=1}^{3} x_{ij} \hfill \text{(Objective)} \nonumber\\[2pt]
&&\text{subject to} \nonumber\\[2pt]
&& \left({\footnotesize\begin{bmatrix} 60 & 2.1 & 1209.6 \\ 45 & 2.22 & 453.6 \\ 30 & 1.25 & 151.2
\end{bmatrix}} \odot X\right)^\top {\footnotesize\begin{bmatrix} 1 \\ 1 \\ 1 \end{bmatrix}} \leq {\footnotesize\begin{bmatrix} 106 \\ 3.5 \\ 1000 \end{bmatrix}} \hfill \text{(Capacity)} \nonumber
\\
&& 0 \;\le\; x_{ij} \;\le\; 1, \quad i,j\in\{1,2,3\} \;\;\; \hfill \text{(Bounds)} \nonumber\\[2pt]
&& 1.2096\,x_{13} = 0.6343\,x_{22}  \;\;\; \hfill \text{(Fairness)} \nonumber\\ [2pt]
&& x_{i3} \;=\; x_{i1},  \; i \in \{1, 2, 3\} \;\; \hfill \text{(Fonthaul Dependency)} \nonumber\\
&& 0.9992\,x_{11} - x_{12}^2 \le 0 \hfill \text{(Latency Dependency)} \nonumber\\
&& 0.9921\,x_{21} - x_{22}^2 \le 0 \hfill \text{(Latency Dependency)} \nonumber\\
&& 0.9733\,x_{31} - x_{32}^2 \le 0 \hfill \text{(Latency Dependency)} \nonumber
\end{IEEEeqnarray}
Where $\odot$ is the Hadamard element-wise product, the program is nonconvex because the constraints in (Latency Dependency) are the hypograph of a convex function. 

To handle the resulting DC (difference-of-convex) constraint, we apply the Convex--Concave Procedure (CCP)~\cite{shen2016disciplined}. CCP linearizes the term $x_{i2}^2$ around a tangent point $a_i$ using its first-order Taylor expansion, namely $x_{i2}^2 \ge a_i^2 + 2a_i(x_{i2}-a_i)=2a_i x_{i2}-a_i^2$. Replacing $x_{i2}^2$ by this affine under-estimator yields the convex surrogate constraint $\alpha_i x_{i1}-(2a_i x_{i2}-a_i^2)\le 0$. This step is conservative: because $2a_i x_{i2}-a_i^2 \le x_{i2}^2$ for all $x_{i2}$, any point satisfying $\alpha_i x_{i1}\le 2a_i x_{i2}-a_i^2$ also satisfies the original constraint $\alpha_i x_{i1}\le x_{i2}^2$, hence feasibility is preserved at every iteration. Starting from an initial feasible point $a_i^{(0)}$, CCP solves the resulting convex program at iteration~$k$ and updates the linearization point as $a_i^{(k+1)}\gets x_{i2}^{(k)}$, repeating until convergence to a stationary (locally optimal) solution. \vspace{-0.2cm}

\subsection{Allocation and Waste Comparison Against Dependency-Agnostic Baselines}

\begin{table}[t]
  \centering
  \caption{Comparison of allocation outcomes under different fairness rules. Waste is measured as the total dependency-violation normalized by total system capacity. Idle is the unallocated capacity normalized by total system capacity}
  \renewcommand{\arraystretch}{1.0}
  {\setlength{\extrarowheight}{6 pt}
  \begin{tabular}{l|c|c|c}
  \hline
  \textbf{Algorithm} &
  \begin{tabular}[c]{@{}c@{}}
    \textbf{Resulting Allocation}\\
    $(N_{\text{PRB}},\ f_{[\text{GHz}]},\ B^{\text{FH-7.2x}}_{[\text{Mbps}]})$
  \end{tabular}
  & \textbf{Waste (\%)} & \textbf{Idle (\%)} \\
  \hline
  DRF & ${\scriptsize\begin{bmatrix}
  15.55 & 0.53 & 313.43\\
  22.24 & 1.10 & 224.14\\
  30.00 & 1.10 & 151.20
\end{bmatrix}}$ & $33.12\%$ & $31.58\%$\\
  MMF & ${\scriptsize\begin{bmatrix}
  38.00 & 1.20 & 424.40\\
  38.00 & 1.20 & 424.40\\
  30.00 & 1.10 & 151.20
\end{bmatrix}}$ & $30.91\%$ & $0\%$\\
  PF & ${\scriptsize\begin{bmatrix}
  33.07 & 1.13 & 666.67\\
  24.80 & 1.22 & 250.00\\
  16.53 & 0.60 & 83.33
\end{bmatrix}}$ & $43.21\%$ & $2.90\%$\\
  \shortstack{Utilitarian\\\tiny\emph{dependency-agnostic}} & ${\scriptsize\begin{bmatrix}
  27.18 & 0.93 & 547.91\\
  29.85 & 1.47 & 300.89\\
  30.00 & 1.10 & 151.20
\end{bmatrix}}$ & $38.21\%$ & $1.71\%$\\
  DDRF & ${\scriptsize\begin{bmatrix}
  18.08 & 1.13 & 364.53\\
  14.98 & 1.28 & 151.02\\
  30.00 & 1.10 & 151.20
\end{bmatrix}}$ & $\mathbf{0\%}$ & $33.92\%$\\
D-Util & ${\scriptsize\begin{bmatrix}
  39.62 & 1.67 & 798.74\\
  4.97 & 0.73 & 50.06\\
  30.00 & 1.10 & 151.20
\end{bmatrix}}$  & $\mathbf{0\%}$ & $2.83\%$\\
  \hline
  \end{tabular}
  }
  \label{tab:wastetable}
\end{table}

We report allocation results (approximated to two decimal digits) together with the percentage of waste and idle resource for the numerical example in Section~\ref{sec:anexample} under several baselines that either ignore inter-resource dependencies or enforce a linear dependency model: (i) DRF, (ii) per-resource MMF (applied independently on each resource), (iii) a generalized PF baseline derived from the MURANES framework~\cite{MURANESPaper} with a scalar satisfaction variable $x_i$ and weights $(1,0,\ldots,0)$, and (iv) a dependency-agnostic utilitarian baseline enforcing a linear proportional dependency across all resources with dependency-aware~\ref{DDRF} and dependency-aware utilitarian~\ref{utilitarian}. As shown in Table~\ref{tab:wastetable}, all baselines incur allocation waste under the nonlinear processing dependency, whereas DDRF and D-Util achieves zero waste. In particular, DDRF satisfies the dependency constraints and fully utilizes at least one congested resource (here, the computing budget). By contrast, DRF saturates none of the resources and can violate the processing-time constraint, yielding PRB allocations that are not effectively processable.

\section{Simulation and evaluation approach}
\label{sec:simulation}
To assess the performance of \textsc{DDRF}, we design an inter-slice scenario. This involves selecting an appropriate demand set, identifying the resources to be allocated across slices, defining representative congestion scenarios, and specifying how resource coupling constraints are enforced.
\vspace*{-0.5 cm}

\subsection{Tenants demands generation}

Demands were derived from Amazon EC2 instance traces~\cite{vantageEC2}, which capture heterogeneous workloads across diverse applications and thus reflect realistic cloud infrastructure patterns relevant to 5G slicing. This diversity makes them a strong benchmark for evaluating allocation mechanisms under varying demand conditions. We selected $23$ demand profiles from different instance families (e.g., general purpose, compute-optimized, memory-optimized), each representing a slice, among these slices, $3$ slices were chosen to be weak (i.e., their demands are low as defined in~\ref{SystemModel}).

Each demand spans four resources: memory (GiB), vCPU, network bandwidth (Gbps), and radio blocks (RBs). The first three are provided directly by the EC2 dataset, while RB demands were synthetically generated. For regular users, RB demands were drawn uniformly from $\mathcal{U}[15,25]$, and for weak users from $\mathcal{U}[1,4]$. All demands were normalized to match the predefined resource capacities: $[17128, 1364, 566.25, 273]$ for (memory, vCPU, bandwidth, RBs), respectively. he choice of $273$ RBs corresponds to a 100\,MHz NR carrier at 30\,kHz SCS, which allows scheduling up to $273$ RBs in the frequency domain per slot~\cite{3gpp_ts_38_104}. In our setting, resources are allocated at the slice level as static, long-term guarantees rather than dynamic per-slot scheduling.
 
\vspace{-0.3 cm}

\subsection{Congestion profiles}

To evaluate the performance of \textsc{DDRF}, we consider varying \textit{congestion profiles}, defined as tuples in $[0,1]^4$ where each coordinate specifies the fraction of the corresponding resource capacity available. Formally, with demand matrix $D$ (columns corresponds to memory, vCPU, bandwidth and RBs respectively) and a congestion profile CP, the available capacity of resource $j$ will be: $c_j = (\sum_{i = 1}^N d_{ij}) \cdot \text{CP}_j $. For example, the congestion profile $(0.3, 0.9, 0.9, 0.9)$ means that only $30\%$ of the total memory demands are available, while $90\%$ of the demands for vCPU, bandwidth, and RBs remain respectively. We impose $14$ congestion profiles to capture both symmetric (e.g., $[0.3, 0.3, 0.3, 0.3]$, $[0.9, 0.9, 0.9, 0.9]$) and asymmetric resource limitations (e.g. $[0.3, 0.8, 0.8, 0.8]$, $[0.8, 0.3, 0.3, 0.8]$).

This design enables stress-testing under both balanced congestion scenarios and highly asymmetric conditions where a single resource becomes the bottleneck.
\vspace{-0.4 cm}

\subsection{Dependency scenarios}

We introduce three types of coupling between resources for each user. In the \emph{linear proportional} case, all resources must be allocated in the same ratio, i.e., $x_{ij} = x_{ik}, \;\; \forall i \in \mathcal{N}, \; j \neq k \in \mathcal{M}$. In the \emph{affine linear} case, dependencies are expressed as constraints of the form $aA_{\text{mem}} + bA_{\text{cpu}} + cA_{\text{bandwidth}} + dA_{rb} + e = 0$, where $A_r$ denotes the allocated amount of resource $r$. In the \emph{polynomial quadratic} case, constraints take the form $aA_{\text{mem}}^\alpha + bA_{\text{cpu}}^\beta + cA_{\text{bandwidth}}^\gamma + dA_{rb}^\eta + e = 0$, with exponents $\{\alpha, \beta, \gamma, \eta\} \in \{0,1,2\}$ and at least one exponent equal to $2$, thus capturing quadratic couplings. 

Based on these definitions, we design three experiments: (i) all couplings follow the linear proportional model, (ii) all couplings follow the affine linear model, and (iii) all couplings follow polynomial quadratic model with $\gamma = 2$ and $\alpha = \beta = \eta = 1$. In all experiments, we consider that each user’s dependency family is $\mathcal{S}_i = \{\mathcal{M}\}, \; \forall \; i \in \mathcal{N}$, thus, all resources belong to a single dependency set and are mutually dependent.
The coefficients $(a, b, c, d, e)$  are derived to align with the demand matrix such that allocating the full demand of each user (i.e., $f_i^{(k)}(\{1\}_{j\in S_i^{(k)}})=0$ (or $\le 0$)) satisfies the coupling constraints. This construction is consistent with our assumption on $\mathcal{F}$.
 All optimization problems were solved using the \texttt{CVXPY}~\cite{diamond2016cvxpy} library with the DCCP extension~\cite{shen2016dccp}.
  \vspace{-0.4 cm}
 \subsection{Baselines}
 
We benchmark \textsc{DDRF} against the following baselines: 
DRF obtained from~\cite{MURANESPaper}, with resource aggregation $s_ix_i, \; X \in \mathbb{R}^N$ and $w = (1, 0 , \cdots, 0)$;
PF obtained from~\cite{MURANESPaper}, which generalizes weighted proportional fairness (PF) to multi-resource settings, where resource aggregation is $x_i, X \in \mathbb{R}^N$ and $w = (1, 0, \cdots, 0)$; 
Mood~\cite{MURANESPaper}, which allocates resources under complete information sharing, here resource aggregation is $PS_ix_i, \; X \in \mathbb{R}^N$ and $PS_i$ is the player satisfaction of user $i$ defined in~\cite{moodvalue} computed on his bottleneck resource with $w = (1, 0, \cdots, 0)$; 
MMF, applied independently on each resource;
and the \emph{dependency-aware utilitarian} solution defined in \eqref{utilitarian}.
 \vspace{-0.3 cm}

\subsection{Effective Satisfaction Formulation}

We formalize the notion of \emph{effective satisfaction}, which quantifies the portion of an allocation that remains usable once tenant dependencies are enforced. This captures the distinction between raw allocations and the subset that yields meaningful utility under inter-resource couplings.

\begin{definition}[Effective Satisfaction Region]
Given a satisfaction matrix $x \in [0,1]^{N \times M}$ and global dependency constraints $\mathcal{F}$, the \emph{effective satisfaction region} is:
\[
\mathcal{E}(X, \mathcal{F}) = \left\{ e \in [0,1]^{N \times M} \;\middle|\; 0 \leq e_{ij} \leq X_{ij}, \; \forall i,j, \;\; e \in \mathcal{F} \right\}
\]
\end{definition}
This set consists of all feasible satisfaction matrices $e$ that (i) do not exceed the raw allocation $x$, and (ii) respect every tenant’s declared dependencies.

\begin{definition}[Effective Satisfaction]
The \emph{effective satisfaction} matrix is defined as:
\[
X^{\text{effective}} = \arg\max_{e \in \mathcal{E}(X, \mathcal{F})} \;\; \sum_{i \in \mathcal{N}} \sum_{j \in \mathcal{M}} e_{ij}.
\]
\end{definition}
Only the dependency-respecting portion of the allocation contributes to utility, ensuring that efficiency is measured in terms of usable satisfaction.

\emph{\textbf{Example.} (Linear Dependency)}
Suppose each tenant has a linear dependency across two resources: $x_{11} = x_{12}$, $x_{21} = x_{22}$. For an allocation
$X =$ {\footnotesize $ \begin{bmatrix} 0.3 & 0.5 \\ 0.2 & 0.7 \end{bmatrix}$}, effective satisfaction enforces equality across each tenant’s resources, yielding
$X^{\text{effective}} = $ {\footnotesize $\begin{bmatrix} 0.3 & 0.3 \\ 0.2 & 0.2 \end{bmatrix}$}.

\emph{\textbf{Example.} (Nonlinear Dependency)}
Consider quadratic couplings: $(a_{11})^2 = a_{12}$ and $(a_{22})^2 = a_{21}$. Given 
$X = $ {\footnotesize $\begin{bmatrix} 0.5 & 0.5 \\ 0.6 & 0.6 \end{bmatrix}$}, projecting onto the feasible region defined by these constraints yields
$X^{\text{effective}} = $ {\footnotesize $\begin{bmatrix} 0.5 & 0.25 \\ 0.36 & 0.6 \end{bmatrix}$}.
 \vspace{-0.5 cm}
\subsection{Performance metrics}
We assess allocations using three complementary metrics:  

\subsubsection{Capacity partitioning}
In multi-resource environments with dependencies, not all allocated resources contribute to utility. We therefore partition system capacity into:
\begin{itemize}
    \item \textit{Used capacity}: the effectively utilized portion,
    $\sum_{i \in \mathcal{N}} \sum_{j \in \mathcal{M}} X^{\text{effective}}_{ij} d_{ij}$.
    \item \textit{Wasted capacity}: resources allocated but unusable due to violated dependencies,
    $\sum_{i \in \mathcal{N}} \sum_{j \in \mathcal{M}} (X_{ij} - X^{\text{effective}}_{ij}) d_{ij}$. 
    \item \textit{Idle capacity}: resources not allocated and remain with InP,
    $\sum_{j \in \mathcal{M}} \left(c_j - \sum_{i \in \mathcal{N}} X_{ij} d_{ij}\right)$.
\end{itemize}

Only the used capacity reflects realizable tenant satisfaction, while wasted and idle capture different forms of inefficiency. we report the fraction of capacity that is \emph{used}, \emph{wasted} (violating dependencies), or \emph{idle} to draw conclusions on the fraction of allocated capacity that was usable by each baseline and how fairness affects efficiency.  
\subsubsection{Cumulative density function (CDF) of the effective satisfaction}
We plot the distribution of realized satisfactions across users, reflecting how heterogeneous demands are met. the smoother and gradually increasing the curve the balanced the allocation was without any jumps (some users given low satisfaction while others are high).

\subsubsection{Fairness}
Fairness is evaluated using \emph{Jain's index}~\cite{jain1984quantitative}, a standard metric in single-resource settings.  
For any vector $z = (z_1,\dots,z_N)$ (either allocations or satisfactions $\{a_{ij}\}_{i = 1}^N$, $\{x_{ij}\}_{i = 1}^N \; \forall j \in \mathcal{M}$), it is defined as
$
J(z) = \frac{\left(\sum_{i=1}^N z_i\right)^2}{N \sum_{i=1}^N z_i^2}.
$
We compute $J$ per resource on the allocation and report the average across all resources between \textsc{DDRF} and dependency-aware utilitarian baseline.


\section{Numerical Results}
\label{sec:results}

We report numerical results for each evaluation metric, detailing the methodology, obtained figures, and  interpretation\footnote{DDRF and simulations'  implementation are available on: \url{https://gitlab.roc.cnam.fr/braikz/dependency-aware-drf.git}.}.

\begin{figure*}
    \centering
    \includegraphics[width=0.8\linewidth]{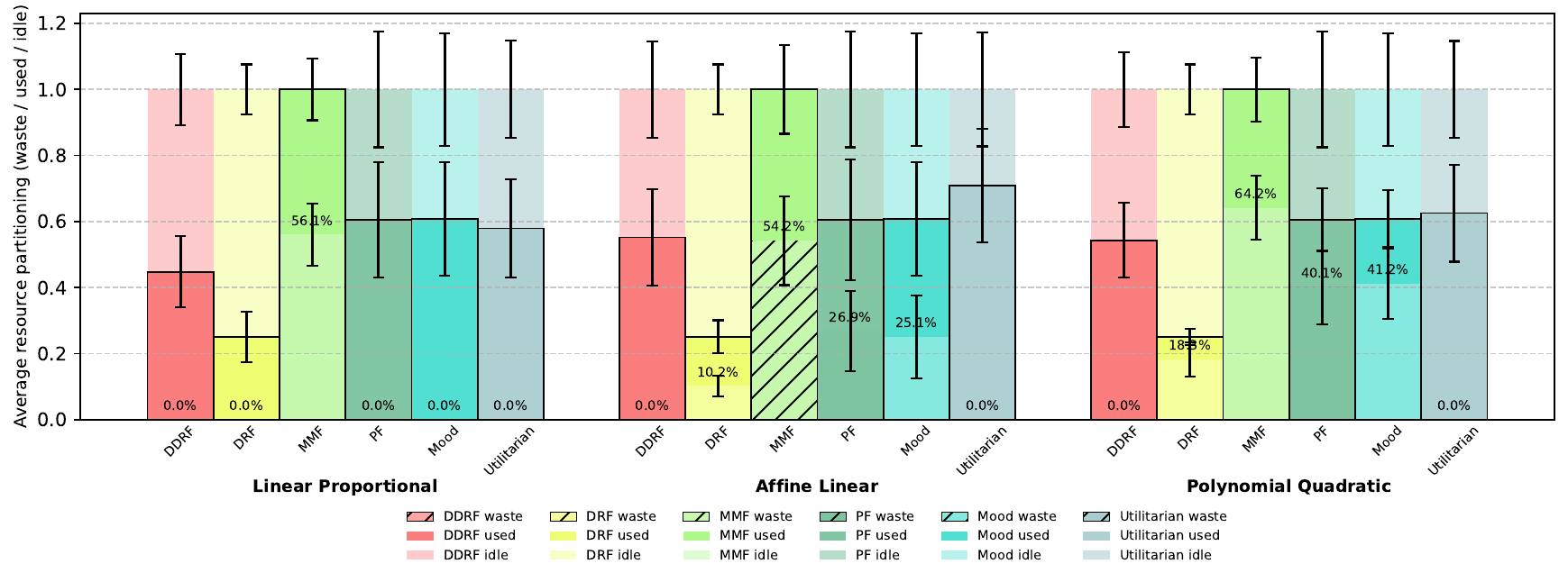}
    \caption{Average partitioning of total resource capacity into wasted, useful, and idle components under different dependency structures across congestion profiles. \vspace*{-0.5cm}}
    \label{fig:capacitydistribution}
\end{figure*}
\vspace{-0.3 cm}

\subsection{Dependency-Aware vs.\ Dependency-Agnostic}
\subsubsection{Resource Partitioning}

Figure~\ref{fig:capacitydistribution} reports the average partitioning of total capacity across congestion profiles. Each bar is divided into three components: \emph{wasted} (hatched, with the percentage labeled), \emph{used} (solid, outlined together with waste), and \emph{idle} (light shade). For each congestion profile, we run all algorithms to obtain the satisfaction matrix, derive the corresponding effective matrix, and then compute the shares of wasted, used, and idle capacity. Results are averaged across profiles to smooth out demand-specific variations.  

Several trends emerge. First, as expected, both \textsc{DDRF} and the Utilitarian approach achieve \emph{zero waste} across all settings (decreasing allocation waste by roughly $60\%$), since both explicitly enforce resource dependencies. In contrast, all dependency-agnostic baselines (DRF, PF, Mood, and MMF) incur non-trivial waste because they ignore coupling. MMF in particular always saturates resources but yields the \emph{highest waste}, as it allocates without regard to dependencies. Second, \textsc{DDRF} consistently improves efficiency compared to DRF: the presence of weak users leads DRF to stall before saturating resources, resulting in idle capacity, whereas DDRF reallocates surplus to other tenants once weak users are fully satisfied (effective Satisfaction rate enhanced from $6.2 \%$ to $60\%$ in the polynomial quadratic case). Interestingly, PF and Mood align
to identical behavior in these cases, confirming their sensitivity to weak tenants. Third, the Utilitarian approach achieves the highest overall efficiency, often allocating nearly all available capacity whenever dependencies permit; however, this comes at the cost of fairness: many tenants are fully satisfied while others are left starved. DDRF, by contrast, balances efficiency and fairness, avoiding starvation while still adapting to the coupling structure. Finally, PF and Mood perform comparatively better under homogeneous congestion (where the same congestion level is applied to all resources), even outperforming all other baselines in the linear proportional case (confirming why they have achieved larger average effective satisfaction rate than utilitarian in the linear proportional dependency case, and why they have larger confidence interval).


\begin{figure*}[t]
    \centering
    \includegraphics[width=0.92\linewidth]{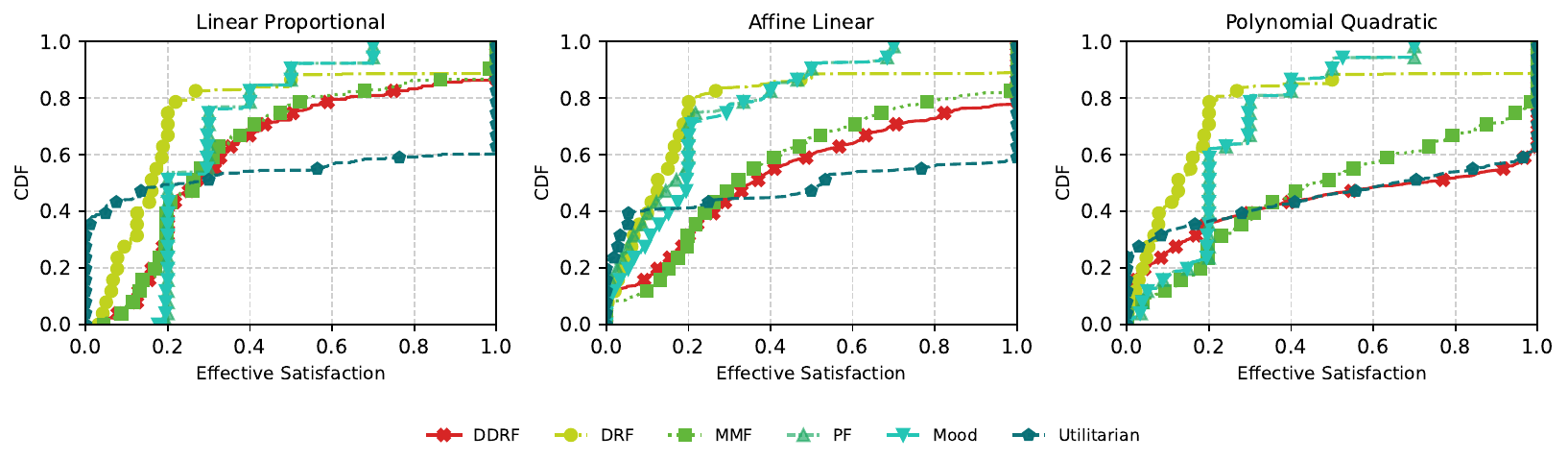}
    \caption{CDF of effective overall satisfaction rate across users and resources, aggregated over all congestion profiles. \vspace*{-0.5cm}}
    \label{fig:overallsatisfaction}
\end{figure*}

\begin{figure*}[t]
    \centering
    \includegraphics[width=0.85\linewidth]{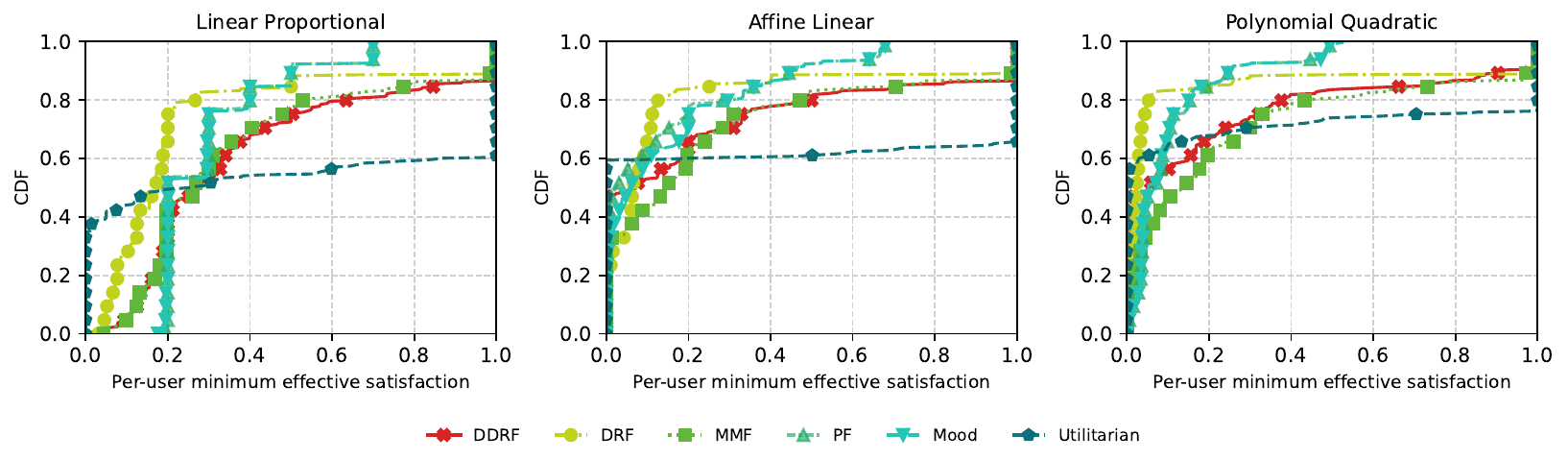}
    \caption{CDF of users’ minimum effective satisfaction rate across resources, aggregated over all congestion profiles. \vspace*{-0.4cm}
}
    \label{fig:minsatisfaction}
\end{figure*}

\subsubsection{Effective Satisfaction Distribution}
Figure~\ref{fig:overallsatisfaction} presents the cumulative distribution function (CDF) of effective satisfaction rates aggregated over all congestion profiles, showing the proportion of user-resource allocations that achieve less than or equal to a given satisfaction level. Complementarily, Figure~\ref{fig:minsatisfaction} reports the CDF of each user’s \emph{minimum} effective satisfaction across resources, aggregated over all profiles, thereby capturing the worst-case satisfaction experienced by each tenant. 

Three key insights. First, \textsc{DDRF} consistently achieves at least the level of satisfaction guaranteed by MMF: in the linear proportional case DDRF and MMF nearly coincide, while in scenarios with generic dependencies DDRF strictly outperforms MMF while enforcing zero waste. Importantly, both DDRF and MMF yield smooth, gradual CDF curves (in overall satisfaction and per-user minimum satisfaction), reflecting balanced satisfaction across users, having this close performance confirms that DDRF is as good as MMF effective satisfaction. Second, the behavior of the other baselines diverges significantly; the Utilitarian approach produces a steep jump in its CDF: many tenants are fully satisfied, but a large fraction receive zero satisfaction, Figure~\ref{fig:minsatisfaction} emphasize this aspect. PF and Mood collapse to identical performance, characterized by vertical growth in their CDFs due to strict satisfaction equalization. Finally, DRF is clustered with PF and Mood in the top-left region of the CDF plots, confirming that most users receive low satisfaction when dependencies are ignored. Taken together, these results show that DDRF combines the fairness balance of MMF with the ability to handle generic dependencies, outperforming all baselines in worst-case user satisfaction without introducing waste.


\begin{figure}[t]
    \centering
    \includegraphics[width=0.8\linewidth]{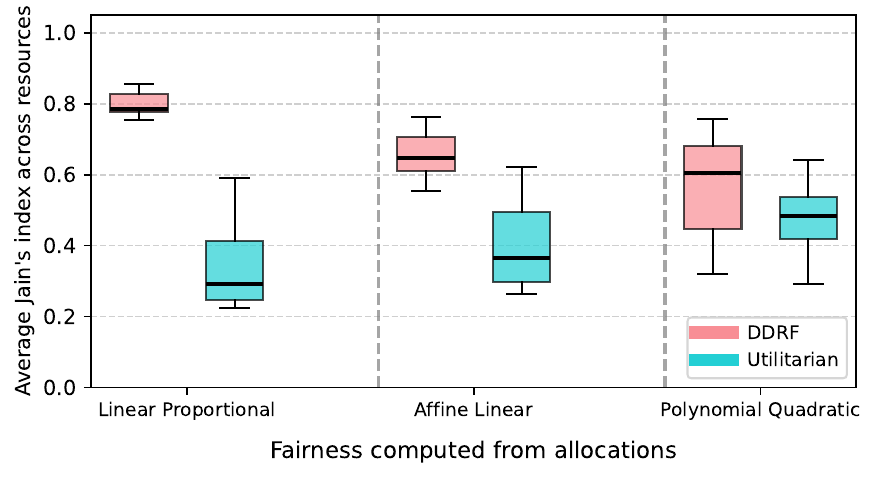}
    \caption{Boxplots of Jain’s fairness index across congestion profiles for DDRF and Utilitarian. We report fairness computed on allocation.}
    \label{fig:fairnessboxplot}
\end{figure}
\vspace{-0.4 cm}

\subsection{Fairness of Dependency-Aware Allocation}

To assess fairness, we quantify it using Jain’s index and compare \textsc{DDRF} against the Utilitarian baseline. Figure~\ref{fig:fairnessboxplot} shows boxplots of the average Jain’s index over resources across all congestion profiles. Fairness is evaluated based on allocations. DDRF, designed as an extension of MMF to multi-resource settings, explicitly enforces fairness at the allocation level. In almost all cases, DDRF achieves higher median fairness than Utilitarian under the three dependency models, it increases fairness by more than $15\%$.
Evidently, \textsc{DDRF} is well suited to allocation-equity-based fairness formulations that account for inter-resource dependencies
\vspace{-0.4 cm}

\subsection{Experimental Evaluation on a vRAN Use Case with Real Resource Dependencies}
The previous results were obtained based on synthetic dependencies, however it is desirable to assess the performance on a realistic use case with real resource dependencies.  

\begin{figure}[t]
    \centering
    \includegraphics[width=0.7\columnwidth]{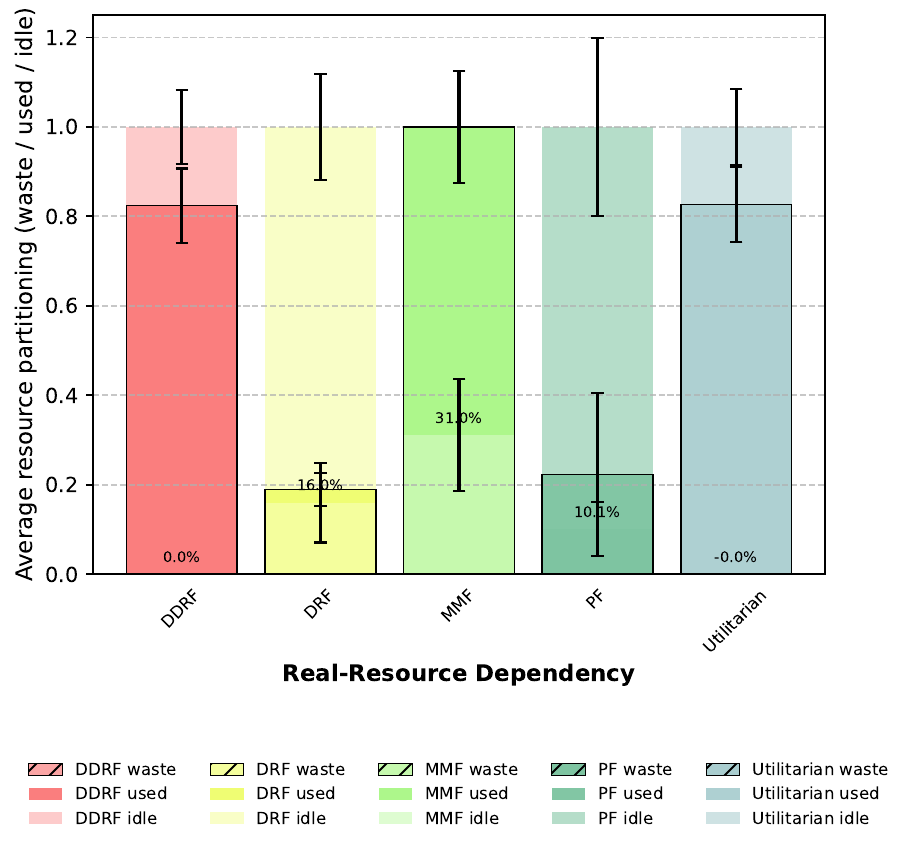}
    \caption{Average partitioning of total resource capacity into wasted, useful, and idle components under different dependency structures across congestion profiles for realistic dependencies. 
    }
    \label{fig:capacitydistributionreal}
\end{figure}

\begin{figure}[t]
    \centering
    \subfloat[]{%
        \includegraphics[width=0.498\columnwidth]{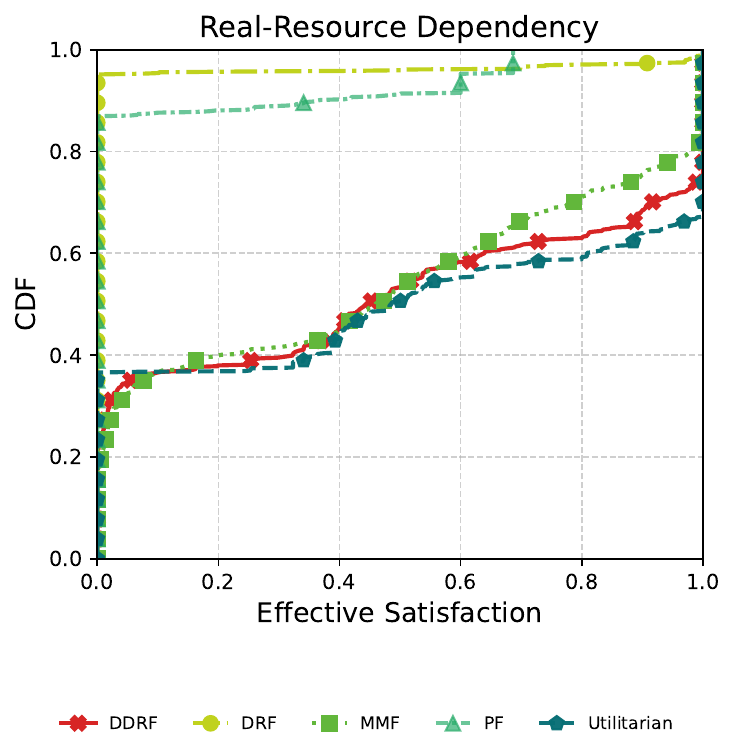}%
        \label{fig:overallsatisfactionreal}%
    }\hfill
    \subfloat[]{%
        \includegraphics[width=0.498\columnwidth]{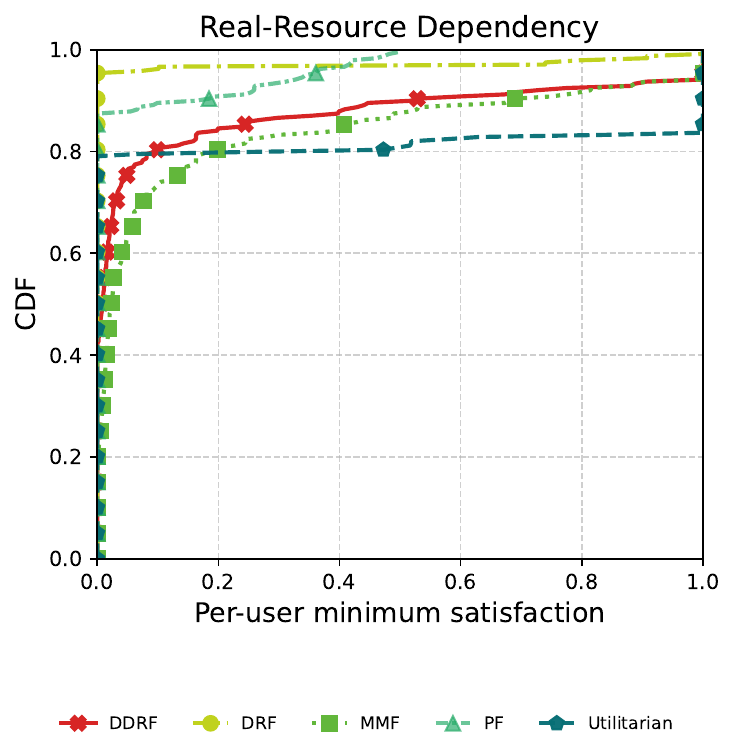}%
        \label{fig:minsatisfactionreal}%
    }

    \caption{CDF of effective (a) overall and (b) minimum satisfaction rate across users and resources, aggregated over all profiles with realistic dependencies.}
    \label{fig:cdf_satisfaction_realistic_deps}
  
\end{figure}

\begin{figure}[t]
    \centering
    \includegraphics[width=0.6\linewidth]{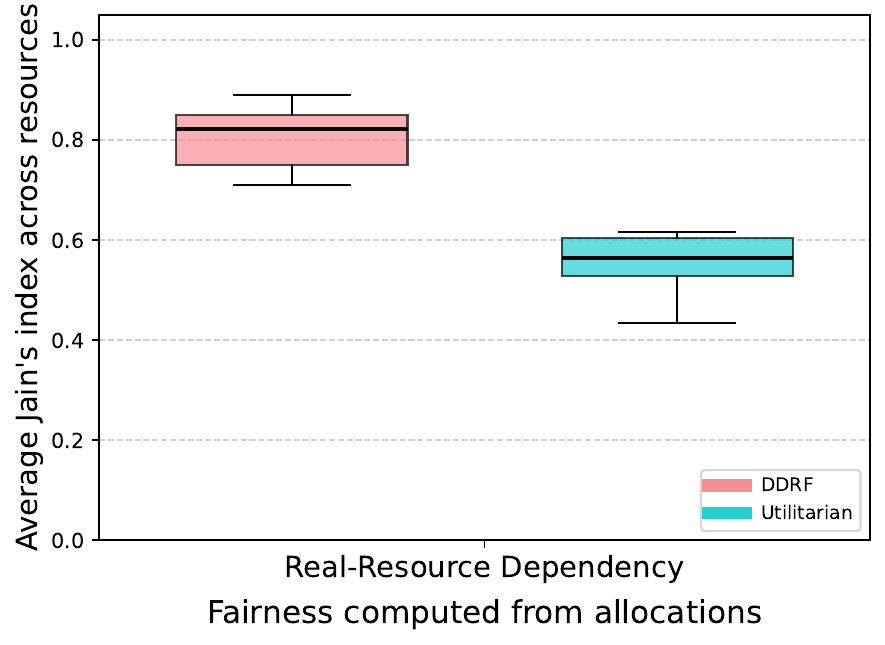}
    \caption{Boxplots of Jain’s fairness index across congestion profiles for DDRF and Utilitarian for realistic dependencies. We report fairness on allocation.}
    \label{fig:fairnessboxplotreal}
\end{figure}

We build a realistic vRAN dependency model by coupling the per-eNB radio load (occupied RBs) to its compute demands (CPU utilization) using measurement-based profiling inspired from~\cite{carlajournal}. We consider an InP operating a vRAN edge cloud co-located with the eNB instances at the same physical site (same room), where multiple virtual co-located eNB instances (any UE can be admitted to any of the eNBs) are provisioned at a non-real-time timescale (seconds to minutes) via the SMO/non-RT RIC through per-eNB budgets for RBs and CPU. To match the adopted profiling regime, each eNB can serve up to 4 UEs ($n \in [1, 4]$) at predefined SINR with $\mathrm{MCS} \in {1,\ldots,27}$, while the occupied RB load varies across slices. For each eNB $i$, we sample $\mathrm{RB}_i \sim \mathcal{U}[1,50]$ for $N-3$ slices and add three ``weak'' slices with $\mathrm{RB}\in\{1,2,3\}$ to increase heterogeneity. We then generate demands $d_i=(d_{i,\mathrm{RB}},d_{i,\mathrm{CPU}},d_{i,n})$ with CPU utilization demand adopted from the regression model taken from \cite{carlajournal},
$
d_{i,\mathrm{CPU}} \;=\; 3.46\, \mathrm{n} + 0.325\,\mathrm{RB}_i + 0.28\,\mathrm{MCS_i} + 26.55
\quad (\mathrm{with}\ \mathrm{MCS_i} \sim \mathcal{U}[1,27]).
$

Capacities are generated from the aggregate demand vector $c=\sum_i d_i$ and a set of pre-defined congestion profiles. We select profiles that create RB-, CPU-, or $n$-bottlenecks while avoiding near-infeasible corners caused by tight CPU baselines under the above equalities (each eNB requires CPU utilization even if no user was admitted and no RBs were assigned). Finally, if the assigned CPU utilization for eNB $i$ is below the baseline term $(0.28\cdot \mathrm{MCS}_i + 26.55)$, we consider all the assigned resources to it as wasted: without sufficient CPU to sustain the baseline processing implied by the selected MCS, any allocated RB/UE budget is effectively unusable at run time, hence it cannot be converted into delivered service.

Similar to the synthetic experiment we report stacked bar plot for capacity partitioning, Figure~\ref{fig:capacitydistributionreal}, CDF for the overall satisfaction, Figure~\ref{fig:overallsatisfactionreal}, CDF for the minimum satisfaction, Figure~\ref{fig:minsatisfactionreal} and allocation based Jain's fairness index between DDRF and utilitarian, Figure~\ref{fig:fairnessboxplotreal}, these figures assert the numerical results obtained with the synthetic experiment. 

Our framework enforces fairness while respecting user-specific inter-resource dependencies. When dependencies are affine, the resulting DDRF program remains convex and can be solved efficiently. In contrast, nonlinear dependencies typically make the problem nonconvex particularly when equality constraints couple linear and nonlinear terms creating challenges in both feasibility enforcement and global optimization. We therefore emphasize that DDRF is a general framework rather than a single fixed solver: the appropriate solution approach depends on the structure of the dependency functions. In practice, one can leverage problem-specific properties via convexification, sequential linearization, conservative relaxations, or carefully designed heuristics, while preserving the fairness logic of DDRF.

Another limitation concerns fairness quantification. The multi-resource allocation literature still lacks a widely accepted global quantitative fairness measure. More importantly, we argue that any meaningful metric must be dependency-aware: assessing fairness purely from allocated shares without accounting for whether those shares are actually usable under the declared inter-resource couplings can be misleading. This is precisely what we observed in the Jain-index boxplots under nonlinear proportional dependencies, where allocation-based fairness did not reflect the effective fairness experienced by tenants. Designing robust dependency-aware fairness metrics is therefore a key direction we leave for future work.

\section{Conclusion and perspectives}
\label{conclusion}

This paper introduced \textit{Dependency-Aware Dominant Resource Fairness (DDRF)}, a generalization of max--min fairness (MMF) to multi-resource allocation under generic inter-resource dependencies. Building on the principles of dominant resource fairness (DRF), DDRF resolves three key shortcomings: it properly accounts for weak users, redefines the bottleneck resource in the presence of dependencies, and reformulates the notion of user satisfaction to reflect dependency-feasible service levels. As a result, DDRF guarantees Pareto efficiency and bottleneck fairness.

The framework is practical: given demand vectors, resource capacities, and dependency constraints, DDRF computes fair allocations in a centralized manner. Extensive experiments show that DDRF maximizes the number of fully satisfied tenants without starving others, strictly outperforms MMF in terms of effective satisfaction, and guarantees resource saturation whenever feasible. Across a broad range of settings - from linear to nonlinear dependencies - DDRF consistently achieved zero waste, improved satisfaction compared to dependency-agnostic approaches, and delivered stronger fairness than dependency-aware baselines. Notably, DDRF increased the effective satisfaction rate by up to $80\%$, improved Jain's fairness by up to $60\%$, and reduced waste by about $60\%$.

Future directions include extending our theoretical guarantees beyond Pareto efficiency and saturation. While we derived structural assumptions on dependency sets that ensure saturation of at least one congested resource and thus Pareto efficiency an important next step is to identify weaker (more general) assumptions under which additional fairness properties can be guaranteed, such as share-incentive, strategy-proofness, and envy-freeness, especially in the presence of nonlinear dependencies. Finally, moving beyond centralized allocation is essential for scalability, motivating the design of distributed implementations suitable for large-scale systems.



\bibliographystyle{IEEEtran} 
\bibliography{references} 

\begin{thebibliography}{10}
\providecommand{\url}[1]{#1}
\csname url@samestyle\endcsname
\providecommand{\newblock}{\relax}
\providecommand{\bibinfo}[2]{#2}
\providecommand{\BIBentrySTDinterwordspacing}{\spaceskip=0pt\relax}
\providecommand{\BIBentryALTinterwordstretchfactor}{4}
\providecommand{\BIBentryALTinterwordspacing}{\spaceskip=\fontdimen2\font plus
\BIBentryALTinterwordstretchfactor\fontdimen3\font minus
  \fontdimen4\font\relax}
\providecommand{\BIBforeignlanguage}[2]{{%
\expandafter\ifx\csname l@#1\endcsname\relax
\typeout{** WARNING: IEEEtran.bst: No hyphenation pattern has been}%
\typeout{** loaded for the language `#1'. Using the pattern for}%
\typeout{** the default language instead.}%
\else
\language=\csname l@#1\endcsname
\fi
#2}}
\providecommand{\BIBdecl}{\relax}
\BIBdecl

\bibitem{chandra2000surplus}
A.~Chandra, M.~Adler, P.~Goyal, and P.~Shenoy, ``Surplus fair scheduling: A
  proportional-share cpu scheduling algorithm for symmetric multiprocessors,''
  in \emph{Proceedings of the 4th Symposium on Operating Systems Design and
  Implementation (OSDI)}.\hskip 1em plus 0.5em minus 0.4em\relax USENIX
  Association, 2000.

\bibitem{boutin2014apollo}
E.~Boutin, J.~Ekanayake, W.~Lin, B.~Shi, J.~Zhou, Z.~Qian, M.~Wu, and L.~Zhou,
  ``Apollo: Scalable and coordinated scheduling for cloud-scale computing,'' in
  \emph{Proceedings of the 11th USENIX Symposium on Operating Systems Design
  and Implementation (OSDI)}.\hskip 1em plus 0.5em minus 0.4em\relax USENIX
  Association, 2014, pp. 285--300.

\bibitem{slicing5G}
P.~Rost, C.~Mannweiler, D.~S. Michalopoulos, C.~Sartori, V.~Sciancalepore,
  N.~Sastry, O.~Holland, S.~Tayade, B.~Han, D.~Bega, D.~Aziz, and H.~Bakker,
  ``Network slicing to enable scalability and flexibility in 5g mobile
  networks,'' \emph{IEEE Communications Magazine}, vol.~55, no.~5, pp. 72--79,
  2017.

\bibitem{proportionalfairness}
F.~Kelly, ``Charging and rate control for elastic traffic,'' \emph{European
  transactions on Telecommunications}, vol.~8, no.~1, pp. 33--37, 1997.

\bibitem{proportionalfairnessequation}
F.~P. Kelly, A.~K. Maulloo, and D.~K.~H. Tan, ``Rate control for communication
  networks: shadow prices, proportional fairness and stability,'' \emph{Journal
  of the Operational Research society}, vol.~49, no.~3, pp. 237--252, 1998.

\bibitem{bertsekasbook}
D.~Bertsekas and R.~Gallager, \emph{Data networks}.\hskip 1em plus 0.5em minus
  0.4em\relax Athena Scientific, 2021.

\bibitem{mo2002fair}
J.~Mo and J.~Walrand, ``Fair end-to-end window-based congestion control,''
  \emph{IEEE/ACM Transactions on networking}, vol.~8, no.~5, pp. 556--567,
  2002.

\bibitem{lan2010axiomatic}
T.~Lan, D.~Kao, M.~Chiang, and A.~Sabharwal, \emph{An axiomatic theory of
  fairness in network resource allocation}.\hskip 1em plus 0.5em minus
  0.4em\relax IEEE, 2010.

\bibitem{DRFpaper}
A.~Ghodsi, M.~Zaharia, B.~Hindman, A.~Konwinski, S.~Shenker, and I.~Stoica,
  ``Dominant resource fairness: Fair allocation of multiple resource types,''
  in \emph{8th USENIX symposium on networked systems design and implementation
  (NSDI 11)}, 2011.

\bibitem{wang2014multi}
W.~Wang, B.~Liang, and B.~Li, ``Multi-resource fair allocation in heterogeneous
  cloud computing systems,'' \emph{IEEE Transactions on Parallel and
  Distributed Systems}, vol.~26, no.~10, pp. 2822--2835, 2014.

\bibitem{joe2013multiresource}
C.~Joe-Wong, S.~Sen, T.~Lan, and M.~Chiang, ``Multiresource allocation:
  Fairness--efficiency tradeoffs in a unifying framework,'' \emph{IEEE/ACM
  Transactions on Networking}, vol.~21, no.~6, pp. 1785--1798, 2013.

\bibitem{MURANESPaper}
F.~Fossati, S.~Moretti, P.~Perny, and S.~Secci, ``Multi-resource allocation for
  network slicing,'' \emph{IEEE/ACM Transactions on Networking}, vol.~28,
  no.~3, pp. 1311--1324, 2020.

\bibitem{bonaldbmf}
\BIBentryALTinterwordspacing
T.~Bonald and J.~Roberts, ``Multi-resource fairness: Objectives, algorithms and
  performance,'' in \emph{Proceedings of the 2015 ACM SIGMETRICS International
  Conference on Measurement and Modeling of Computer Systems}, ser. SIGMETRICS
  '15.\hskip 1em plus 0.5em minus 0.4em\relax New York, NY, USA: Association
  for Computing Machinery, 2015, p. 31–42. [Online]. Available:
  \url{https://doi.org/10.1145/2745844.2745869}
\BIBentrySTDinterwordspacing

\bibitem{hug}
M.~Chowdhury, Z.~Liu, A.~Ghodsi, and I.~Stoica, ``Hug: multi-resource fairness
  for correlated and elastic demands,'' in \emph{Proceedings of the 13th Usenix
  Conference on Networked Systems Design and Implementation}, ser.
  NSDI'16.\hskip 1em plus 0.5em minus 0.4em\relax USA: USENIX Association,
  2016, p. 407–424.

\bibitem{surveyresourceallocation}
P.~Poullie, T.~Bocek, and B.~Stiller, ``A survey of the state-of-the-art in
  fair multi-resource allocations for data centers,'' \emph{IEEE Transactions
  on Network and Service Management}, vol.~15, no.~1, pp. 169--183, 2018.

\bibitem{polese2023understanding}
M.~Polese, L.~Bonati, S.~D’oro, S.~Basagni, and T.~Melodia, ``Understanding
  o-ran: Architecture, interfaces, algorithms, security, and research
  challenges,'' \emph{IEEE Communications Surveys \& Tutorials}, vol.~25,
  no.~2, pp. 1376--1411, 2023.

\bibitem{khatibi2018modelling}
S.~Khatibi, K.~Shah, and M.~Roshdi, ``Modelling of computational resources for
  5g ran,'' in \emph{2018 European Conference on Networks and Communications
  (EuCNC)}.\hskip 1em plus 0.5em minus 0.4em\relax IEEE, 2018, pp. 1--5.

\bibitem{hojeij2025flexible}
H.~Hojeij, G.~I. Ricardo, M.~Sharara, S.~Hoteit, V.~V{\`e}que, and S.~Secci,
  ``On flexible association and placement in disaggregated ran designs,''
  \emph{Computer Communications}, vol. 238, p. 108166, 2025.

\bibitem{xiao2020can}
Y.~Xiao, J.~Zhang, and Y.~Ji, ``Can fine-grained functional split benefit to
  the converged optical-wireless access networks in 5g and beyond?'' \emph{IEEE
  Transactions on Network and Service Management}, vol.~17, no.~3, pp.
  1774--1787, 2020.

\bibitem{3gpp.38.214}
\BIBentryALTinterwordspacing
\emph{{TS 38.214: NR; Physical layer procedures for data}}, 3rd Generation
  Partnership Project (3GPP) Std. TS 38.214, 2023, release 17, v17.6.0.
  [Online]. Available:
  \url{https://www.3gpp.org/ftp/Specs/archive/38_series/38.214/}
\BIBentrySTDinterwordspacing

\bibitem{marshall1979inequalities}
A.~W. Marshall, I.~Olkin, and B.~C. Arnold, \emph{Inequalities: theory of
  majorization and its applications}.\hskip 1em plus 0.5em minus 0.4em\relax
  Springer, 1979.

\bibitem{zaharia2010delay}
M.~Zaharia, D.~Borthakur, J.~Sen~Sarma, K.~Elmeleegy, S.~Shenker, and
  I.~Stoica, ``Delay scheduling: a simple technique for achieving locality and
  fairness in cluster scheduling,'' in \emph{Proceedings of the 5th European
  conference on Computer systems}, 2010, pp. 265--278.

\bibitem{popa2012faircloud}
L.~Popa, G.~Kumar, M.~Chowdhury, A.~Krishnamurthy, S.~Ratnasamy, and I.~Stoica,
  ``Faircloud: Sharing the network in cloud computing,'' in \emph{Proceedings
  of the ACM SIGCOMM 2012 conference on Applications, technologies,
  architectures, and protocols for computer communication}, 2012, pp. 187--198.

\bibitem{rawls1971theory}
J.~Rawls, ``A theory of justice, harvard,'' \emph{Press, Cambridge}, 1971.

\bibitem{ogryczak2014fair}
W.~Ogryczak, H.~Luss, M.~Pi{\'o}ro, D.~Nace, and A.~Tomaszewski, ``Fair
  optimization and networks: A survey,'' \emph{Journal of Applied Mathematics},
  vol. 2014, no.~1, p. 612018, 2014.

\bibitem{mmfnsdi}
\BIBentryALTinterwordspacing
P.~Namyar, B.~Arzani, S.~Kandula, S.~Segarra, D.~Crankshaw, U.~Krishnaswamy,
  R.~Govindan, and H.~Raj, ``Solving {Max-Min} fair resource allocations
  quickly on large graphs,'' in \emph{21st USENIX Symposium on Networked
  Systems Design and Implementation (NSDI 24)}.\hskip 1em plus 0.5em minus
  0.4em\relax Santa Clara, CA: USENIX Association, Apr. 2024, pp. 1937--1958.
  [Online]. Available:
  \url{https://www.usenix.org/conference/nsdi24/presentation/namyar-solving}
\BIBentrySTDinterwordspacing

\bibitem{moulin2002axiomatic}
H.~Moulin, ``Axiomatic cost and surplus sharing,'' \emph{Handbook of social
  choice and welfare}, vol.~1, pp. 289--357, 2002.

\bibitem{moodvalue}
F.~Fossati, S.~Hoteit, S.~Moretti, and S.~Secci, ``Fair resource allocation in
  systems with complete information sharing,'' \emph{IEEE/ACM Transactions on
  Networking}, vol.~26, no.~6, pp. 2801--2814, 2018.

\bibitem{parkes2015beyond}
D.~C. Parkes, A.~D. Procaccia, and N.~Shah, ``Beyond dominant resource
  fairness: Extensions, limitations, and indivisibilities,'' \emph{ACM
  Transactions on Economics and Computation (TEAC)}, vol.~3, no.~1, pp. 1--22,
  2015.

\bibitem{wang2014dominant}
W.~Wang, B.~Li, and B.~Liang, ``Dominant resource fairness in cloud computing
  systems with heterogeneous servers,'' in \emph{IEEE INFOCOM 2014-IEEE
  Conference on Computer Communications}.\hskip 1em plus 0.5em minus
  0.4em\relax IEEE, 2014, pp. 583--591.

\bibitem{chowdhury2016hug}
M.~Chowdhury, Z.~Liu, A.~Ghodsi, and I.~Stoica,
  ``$\{$HUG$\}$:$\{$Multi-Resource$\}$ fairness for correlated and elastic
  demands,'' in \emph{13th USENIX symposium on networked systems design and
  implementation (NSDI 16)}, 2016, pp. 407--424.

\bibitem{BMFpaper}
\BIBentryALTinterwordspacing
T.~Bonald and J.~Roberts, ``Multi-resource fairness: Objectives, algorithms and
  performance,'' \emph{SIGMETRICS Perform. Eval. Rev.}, vol.~43, no.~1, p.
  31–42, Jun. 2015. [Online]. Available:
  \url{https://doi.org/10.1145/2796314.2745869}
\BIBentrySTDinterwordspacing

\bibitem{dalton1920measurement}
H.~Dalton, ``The measurement of the inequality of incomes,'' \emph{The Economic
  Journal}, vol.~30, no. 119, pp. 348--361, 1920.

\bibitem{shen2016disciplined}
\BIBentryALTinterwordspacing
X.~Shen, S.~Diamond, Y.~Gu, and S.~Boyd, ``Disciplined convex-concave
  programming,'' \emph{2016 IEEE 55th Conference on Decision and Control
  (CDC)}, pp. 1009--1014, 2016. [Online]. Available:
  \url{https://stanford.edu/~boyd/papers/dccp.html}
\BIBentrySTDinterwordspacing

\bibitem{vantageEC2}
Vantage, ``Ec2instances.info - easy amazon ec2 instance comparison,''
  \url{https://instances.vantage.sh/}.

\bibitem{3gpp_ts_38_104}
\BIBentryALTinterwordspacing
{3GPP}, ``{TS 38.104: NR; Base Station (BS) radio transmission and
  reception},'' 3rd Generation Partnership Project (3GPP), Tech. Rep., 2023,
  release 17, v17.9.0. [Online]. Available:
  \url{https://www.etsi.org/deliver/etsi_ts/138100_138199/138104/17.09.00_60/ts_138104v170900p.pdf}
\BIBentrySTDinterwordspacing

\bibitem{diamond2016cvxpy}
S.~Diamond and S.~Boyd, ``{CVXPY}: A {P}ython-embedded modeling language for
  convex optimization,'' \emph{Journal of Machine Learning Research}, vol.~17,
  no.~83, pp. 1--5, 2016.

\bibitem{shen2016dccp}
X.~Shen, S.~Diamond, Y.~Gu, and S.~Boyd, ``Disciplined convex‐concave
  programming,'' \emph{arXiv preprint arXiv:1604.02639}, 2016.

\bibitem{jain1984quantitative}
R.~K. Jain, D.-M.~W. Chiu, W.~R. Hawe \emph{et~al.}, ``A quantitative measure
  of fairness and discrimination,'' \emph{Eastern Research Laboratory, Digital
  Equipment Corporation, Hudson, MA}, vol.~21, no.~1, pp. 2022--2023, 1984.

\bibitem{carlajournal}
S.~Pramanik, A.~Ksentini, and C.~F. Chiasserini, ``Cost-efficient slicing in
  virtual radio access networks,'' \emph{Computer Communications}, vol. 209,
  pp. 349--358, 2023.

\end{thebibliography}

\appendices

\section{Assumptions on Dependencies for Saturation}
\label{appendixA}

We now specify the assumptions on $\mathcal{F}$ that guarantee the Pareto efficiency of~\eqref{DDRF} and~\eqref{utilitarian}. Consider the dependency-aware multi-resource problem $(D,C,\mathcal{F})$ and assume that at least one resource is congested. Let $X^\star$ be an optimal solution of~\eqref{utilitarian} or~\eqref{DDRF}.

\textit{Local regularity and local improving feasible direction.}
For any tenant $i \in \mathcal{N}$ such that there exists an active congested bottleneck resource $j_i' \in \argmax_{j \in \mathcal C:\, y_{ij}=1} \frac{d_{ij}}{c_j}$, we impose the following assumptions on the dependency constraints $f_i^{(k)}$ of tenant $i$ that are active at $X^\star$, i.e., $f_i^{(k)}(X^\star)=0$, and whose dependency sets $S_i^{(k)} \subseteq \mathcal S_i$ contain $j'$.

Define $\mathcal A_i(j') \triangleq \{\, k \in \{1,\dots,K_i\} : j' \in S_i^{(k)} \text{ and } f_i^{(k)}(X^\star)=0 \,\}$, the set of active dependency constraints of tenant $i$ that contain $j'$ and active at $X^\star$. We distinguish two cases according to the cardinality of $\mathcal A_i(j')$.

\begin{itemize}
\item \textit{Case (i) (unique active constraint).}
Suppose that $\mathcal A_i(j')=\{k\}$. Then we assume:
(1) \textit{(Local regularity)} The function $f_i^{(k)}$ is $C^1$ in a neighborhood of $X^\star$ with respect to the variables $(x_{ir})_{r\in S_i^{(k)}}$. 
(2) \textit{(Improving-direction inequality)} Define $g_r \triangleq \frac{\partial f_i^{(k)}}{\partial x_{ir}}(X^\star)$ for all $r\in S_i^{(k)}$. Then:

$
g_{j'}
\left(
\sum_{r\in S_i^{(k)}\setminus\{j'\}} g_r
\right)
<
\sum_{r\in S_i^{(k)}\setminus\{j'\}} g_r^2.
$
\vspace{0.2 cm}

\item \textit{Case (ii) (multiple active constraints).}
Suppose that $|\mathcal A_i(j')| \ge 2$. Then we assume:
(1) \textit{(Local regularity)} For every $k \in \mathcal A_i(j')$, the function $f_i^{(k)}$ is $C^1$ in a neighborhood of $X^\star$.
(2) \textit{(Joint improving-direction condition)} Let $U_i(j') \triangleq \bigcup_{k\in\mathcal A_i(j')} S_i^{(k)}$ be the set of coordinates of tenant $i$ allowed to vary. Define the Jacobian matrix $J_{\mathrm{act}}(X^\star)\in\mathbb R^{|\mathcal A_i(j')|\times |U_i(j')|}$ by:

$
[J_{\mathrm{act}}(X^\star)]_{k,r}
\triangleq
\frac{\partial f_i^{(k)}}{\partial x_{ir}}(X^\star),
$
\vspace{0.2 cm}

for $k\in\mathcal A_i(j')$ and $r\in U_i(j')$, with the convention that $\frac{\partial f_i^{(k)}}{\partial x_{ir}}(X^\star)=0$ whenever $r\notin S_i^{(k)}$. We assume that there exists a direction $\Delta_i\in\mathbb R^{|U_i(j')|}$ such that $\Delta_{ij'}>0$, $J_{\mathrm{act}}(X^\star)\Delta_i=0$, and $\mathbf 1^\top \Delta_i>0$, and moreover $\nabla f_i^{(k)}(X^\star)^\top \Delta_i \le 0$ for every active inequality constraint $k\in\mathcal A_i(j')$.
\end{itemize}

\section{Proof of Lemma~\ref{lemma1}}
\label{appendixB}
\begin{proof}
We prove saturation (hence Pareto efficiency) by contradiction. Without loss of  generality, we consider the minimal case where there is a single congested resource $j'\in\mathcal M$. 
Let $X^\star$ be an optimal solution of \eqref{utilitarian}. We show that $j'$ must be saturated, i.e., $\sum_{i=1}^N d_{ij'}x_{ij'}^\star=c_{j'}$. Assume for contradiction that $j'$ is not saturated: $\sum_{i=1}^N d_{ij'}x_{ij'}^\star<c_{j'}$. Let the excess be $s_{j'}\triangleq c_{j'}-\sum_{i=1}^N d_{ij'}x_{ij'}^\star>0$. Since $j'$ is congested in the sense $\sum_i d_{ij'}>c_{j'}$, feasibility implies that there exists a tenant $i'\in\mathcal N$ with $x_{i'j'}^\star<1$.

We construct a perturbation supported on tenant $i'$ only. The key intuition is the first-order Taylor expansion of each $C^1$ dependency constraint $f_{i'}^{(k)}$ around $X^\star$: for any direction $\Delta x_{i'}$ and any $k\in\mathcal A_{i'}(j')$, we have:

$f_{i'}^{(k)}(X^\star+\varepsilon\Delta X)=f_{i'}^{(k)}(X^\star)+\varepsilon\,\nabla f_{i'}^{(k)}(X^\star)^\top\Delta x_{i'}+o(\varepsilon)$\\
as $\varepsilon\to 0$, where $o(\varepsilon)$ denotes a remainder satisfying $o(\varepsilon)/\varepsilon\to 0$ as $\varepsilon\to 0$. Hence an active equality constraint $f_{i'}^{(k)}(X^\star)=0$ is preserved to first order by imposing $\nabla f_{i'}^{(k)}(X^\star)^\top\Delta x_{i'}=0$.

Fix the pair $(i',j')$ given by the assumptions of Lemma~\ref{lemma1}. If $|\mathcal A_{i'}(j')|=1$, write $\mathcal A_{i'}(j')=\{k'\}$, let $S\triangleq S_{i'}^{(k')}$, let $J\triangleq S\setminus\{j'\}$, and define $g_r\triangleq \partial f_{i'}^{(k')}/\partial x_{i'r}(X^\star)$ for all $r\in S$. Define a direction $\Delta x_{i'}$ by $\Delta x_{i'j'}=1$, $\Delta x_{i'r}=-(g_{j'}g_r)/\sum_{\ell\in J} g_\ell^2$ for $r\in J$, and $\Delta x_{i'r}=0$ for $r\notin S$. Then $\nabla f_{i'}^{(k')}(X^\star)^\top\Delta x_{i'}=g_{j'}+\sum_{r\in J}g_r\Delta x_{i'r}=0$, hence the unique active dependency constraint is preserved to first order. Moreover, the utility gain along this direction is $\mathbf 1^\top\Delta x_{i'}=1-(g_{j'}\sum_{r\in J}g_r)/\sum_{r\in J}g_r^2$, which is strictly positive by the inequality in Case (i) (2) in appendix~\ref{appendixA}. 

By $C^1$-regularity,  $\nabla f_{i'}^{(k')}(X^\star)^\top\Delta x_{i'}=0$ ensures that the active dependency constraints are preserved to first order. To enforce them \emph{exactly} (not only to first order), we invoke the \textit{implicit function theorem}: since $\sum_{\ell\in J} g_\ell^2>0$, there exists $r_0\in J$ with $\partial f_{i'}^{(k')}/\partial x_{i'r_0}(X^\star)\neq 0$, so for $\varepsilon>0$ small enough one can choose $x_{i'r_0}$ as a $C^1$ function of the remaining perturbed coordinates (including $x_{i'j'}$) so that $f_{i'}^{(k')}(X^\star+\varepsilon\Delta X)=0$ holds exactly. 

We then choose any $\varepsilon\in(0,\bar\varepsilon)$ such that
$\bar\varepsilon \triangleq \min\Bigl\{\frac{s_{j'}}{d_{i'j'}}, \;
\min_{r\in S:\,\Delta x_{i'r}>0}\frac{1-x_{i'r}^\star}{\Delta x_{i'r}},
\;
\min_{r\in S:\,\Delta x_{i'r}<0}\frac{x_{i'r}^\star}{-\Delta x_{i'r}}
\Bigr\}$ Since $\Delta x_{i'j'}=1$, this guarantees that $0\le x_{i'r}^\star+\varepsilon\Delta x_{i'r}\le 1$ for every $r\in S$. Moreover,
$\sum_i d_{ij'}(x_{ij'}^\star+\varepsilon\Delta x_{ij'})
=
\sum_i d_{ij'}x_{ij'}^\star+d_{i'j'}\varepsilon
\le c_{j'}$.
Thus the bound constraints on the perturbed coordinates and the capacity constraint on $j'$ remain satisfied.

Finally, the objective strictly increases: $\sum_{i,j}(x_{ij}^\star+\varepsilon\Delta x_{ij})=\sum_{i,j}x_{ij}^\star+\varepsilon\,\mathbf 1^\top\Delta x_{i'}$, and $\mathbf 1^\top\Delta x_{i'}>0$, hence the objective is larger for any such $\varepsilon>0$. This contradicts optimality of $X^\star$. Therefore $\sum_{i=1}^N d_{ij'}x_{ij'}^\star=c_{j'}$, i.e., at least one congested resource is saturated at the utilitarian optimum and Pareto efficiency immediately follows.

\end{proof}

Note: Inactive inequality constraints are those with $f_i^{(k)}(X^\star)<0$. For such constraints, continuity at $X^\star$ of $f_i^{(k)}$ for $k \in \mathcal{A}_i(j')$ is enough: sufficiently small perturbations preserve strict feasibility. Thus only active constraints need to be enforced explicitly.

\section{Proof of Theorem~\ref{theorem1}}
\label{appendixC}

\begin{proof}

We argue by contradiction and, without loss of generality, consider the minimal nontrivial case of a single congested resource $j'\in\mathcal M$ with exactly two active tenants $i',i''$ on it.

Let $X^\star$ be an optimal solution of \eqref{DDRF}, and suppose that $j'$ is not saturated, i.e.,
$\sum_{i=1}^N d_{ij'}x_{ij'}^\star<c_{j'}$.
Set
$s_{j'}\triangleq c_{j'}-\sum_i d_{ij'}x_{ij'}^\star>0$.
Since $i'$ and $i''$ are active on $j'$, the fairness equality gives
$\mu_{i'j'}x_{i'j'}^\star=\mu_{i''j'}x_{i''j'}^\star$.
We therefore choose the perturbation on the bottleneck coordinate $j'$ as
$x_{i'j'}^\star \mapsto x_{i'j'}^\star+\varepsilon/\mu_{i'j'}$
and
$x_{i''j'}^\star \mapsto x_{i''j'}^\star+\varepsilon/\mu_{i''j'}$,
so that the fairness equality on $j'$ is preserved exactly for every $\varepsilon>0$.

By the assumptions of Appendix~\ref{appendixA}, for each tenant $t\in\{i',i''\}$ there exists a tenant-wise feasible direction $\Delta x_t$ supported on the coordinates coupled to $j'$ such that $\Delta x_{tj'}=1/\mu_{tj'}$, the active dependency constraints containing $j'$ are preserved to first order (and exactly after the same implicit-function as in Lemma~\ref{lemma1}), and $\mathbf 1^\top \Delta x_t>0$. Let $S_t$ denote the set of perturbed coordinates of tenant $t$. We then choose any $\varepsilon\in(0,\bar\varepsilon)$ such that:

$
\bar\varepsilon = \min\Biggl\{
\frac{s_{j'}}{d_{i'j'}/\mu_{i'j'}+d_{i''j'}/\mu_{i''j'}},
\min_{t\in\{i',i''\}} \\ \min_{r\in S_t:\,\Delta x_{tr}>0}\frac{1-x_{tr}^\star}{\Delta x_{tr}},
\min_{t\in\{i',i''\}}\min_{r\in S_t:\,\Delta x_{tr}<0}\frac{x_{tr}^\star}{-\Delta x_{tr}}
\Biggr\}.
$

The remainder of the proof then follows exactly as in Lemma~\ref{lemma1}.
\end{proof}

\section{Proof of Theorem~\ref{theorem2}}
\label{AppendixD}
 
\subsection{Case 1: All Users active, some have bottleneck on a non-congested resource  }
\begin{proof}
Assume linear dependency and that the equalization constraints hold for all users (Active case). With $\alpha_i^{\mathcal C}=1/\mu_i^{\mathcal C}$, $\alpha_i=1/\mu_i$,  $\mu_i^{\mathcal C}x_i=t^{\mathcal C}\ \Rightarrow\ x_i=t^{\mathcal C}\alpha_i^{\mathcal C},
\mu_i x_i=t\ \Rightarrow\ x_i=t\alpha_i.$ \\
Moreover, $x_i\le 1$ implies:
$t^{\mathcal C}\le \min_{i\in\mathcal N}\mu_i^{\mathcal C},
t\le \min_{i\in\mathcal N}\mu_i.$\\
For $j\in\mathcal M_0$:

$\sum_i d_{ij}x_i \le c_j^{\mathcal C}
\ \Rightarrow\
t^{\mathcal C}\sum_i \alpha_i^{\mathcal C}d_{ij}\le c_j^{\mathcal C}
\ \Rightarrow\
t^{\mathcal C}\le \frac{c_j^{\mathcal C}}{\sum_i \alpha_i^{\mathcal C}d_{ij}}.$\\
Hence, $t^{\mathcal C}
=\min\!\left\{
\min_{i\in\mathcal N}\mu_i^{\mathcal C},\ 
\min_{j\in\mathcal M_0}\frac{c_j^{\mathcal C}}{\sum_i \alpha_i^{\mathcal C}d_{ij}}
\right\}.$\\
Therefore, 
$\sum_i x_i^{\mathrm{DDRF}}
=t^{\mathcal C}\sum_i \alpha_i^{\mathcal C}
=\min\!\left\{
\Big(\min_{i}\mu_i^{\mathcal C}\Big)\sum_i \alpha_i^{\mathcal C},\ 
\min_{j\in\mathcal M_0}\frac{c_j^{\mathcal C}\sum_i \alpha_i^{\mathcal C}}{\sum_i \alpha_i^{\mathcal C}d_{ij}}
\right\}$

$=\min_{j\in\mathcal M_0}\frac{c_j^{\mathcal C}}{M_1(\alpha^{\mathcal C};d_{\cdot j})}
=\frac{c_{j^\star}^{\mathcal C}}{M_1(\alpha^{\mathcal C};d_{\cdot j^\star})},$

where $j^\star\in\arg\min_{j\in\mathcal M_0}\frac{M_1(\alpha^{\mathcal C};d_{\cdot j})}{c_j^{\mathcal C}}$ (with $c_0^{\mathcal C}=(\min_i\mu_i^{\mathcal C})\sum_i\alpha_i^{\mathcal C}$, $d_{\cdot 0}=\mathbf 1$). \\
Similarly,
$t
=\min\!\left\{
\min_{i\in\mathcal N}\mu_i,\ 
\min_{j\in\mathcal M_0}\frac{c_j}{\sum_i \alpha_i d_{ij}}
\right\},$

$\sum_i x_i^{\mathrm{DRF}}
=t\sum_i \alpha_i
=\min_{j\in\mathcal M_0}\frac{c_j}{M_1(\alpha;d_{\cdot j})}
=\frac{c_{j'}}{M_1(\alpha;d_{\cdot j'})},$

where $j'\in\arg\min_{j\in\mathcal M_0}\frac{M_1(\alpha;d_{\cdot j})}{c_j}$ (with $c_0=(\min_i\mu_i)\sum_i\alpha_i$).\\
Thus:
$\sum_i x_i^{\mathrm{DDRF}}\ge \sum_i x_i^{\mathrm{DRF}}
\iff
\frac{c_{j^\star}^{\mathcal C}}{M_1(\alpha^{\mathcal C};d_{\cdot j^\star})}
\ge
\frac{c_{j'}}{M_1(\alpha;d_{\cdot j'})}$

$\iff
\frac{M_1(\alpha^{\mathcal C};d_{\cdot j^\star})}{c_{j^\star}^{\mathcal C}}
\le
\frac{M_1(\alpha;d_{\cdot j'})}{c_{j'}}.$

(i) If $\dfrac{M_1(\alpha^{\mathcal C};d_{\cdot j^\star})}{c_{j^\star}^{\mathcal C}}\le \dfrac{M_1(\alpha;d_{\cdot j'})}{c_{j'}}$, then $\sum_i x_i^{\mathrm{DDRF}}\ge \sum_i x_i^{\mathrm{DRF}}$.  

(ii) If $\dfrac{M_1(\alpha^{\mathcal C};d_{\cdot j^\star})}{c_{j^\star}^{\mathcal C}} > \dfrac{M_1(\alpha;d_{\cdot j'})}{c_{j'}}$, then $\sum_i x_i^{\mathrm{DDRF}}< \sum_i x_i^{\mathrm{DRF}}$.

\subsection{Case 2: At least one User weak, some have bottleneck on a non-congested resource }
In DDRF, $x_i=1$ for all $i\in\mathcal W$. For $i\in\mathcal A$.\\
The equalization constraints give
$\mu_i^{\mathcal C}x_i=t^{\mathcal C}\ \Rightarrow\ x_i=t^{\mathcal C}\alpha_i^{\mathcal C},\qquad i\in\mathcal A,$\\
Hence $x_i\le 1$ implies $t^{\mathcal C}\le \min_{i\in\mathcal A}\mu_i^{\mathcal C}$. For $j\in\mathcal M$,\\
$\sum_{i\in\mathcal N} d_{ij}x_i\le c_j
\ \Rightarrow\
\sum_{i\in\mathcal A} d_{ij}x_i \le c_j-\sum_{i\in\mathcal W}d_{ij}
\ \Rightarrow\
t^{\mathcal C}\sum_{i\in\mathcal A}\alpha_i^{\mathcal C}d_{ij}\le \tilde c_j,$

where $\tilde c_j\triangleq c_j-\sum_{i\in\mathcal W}d_{ij}$. With $\tilde c_0\triangleq (\min_{i\in\mathcal A}\mu_i^{\mathcal C})\sum_{i\in\mathcal A}\alpha_i^{\mathcal C}$ and $\mathcal M_0=\mathcal M\cup\{0\}$ (and $d_{i0}=1$).\\
This yields:
$t^{\mathcal C}=\min_{j\in\mathcal M_0}\frac{\tilde c_j}{\sum_{i\in\mathcal A}\alpha_i^{\mathcal C}d_{ij}},$\\
$\sum_{i\in\mathcal N}x_i^{\mathrm{DDRF}}
=
|\mathcal W|+t^{\mathcal C}\sum_{i\in\mathcal A}\alpha_i^{\mathcal C}
=
|\mathcal W|+\min_{j\in\mathcal M_0}\frac{\tilde c_j}{M_1(\alpha_{\mathcal A}^{\mathcal C};d_{\mathcal A,\cdot j})}$
$=
|\mathcal W|+\frac{\tilde c_{\tilde j^\star}}{M_1(\alpha_{\mathcal A}^{\mathcal C};d_{\mathcal A,\cdot \tilde j^\star})},$

where $\tilde j^\star\in\arg\min_{j\in\mathcal M_0}\frac{M_1(\alpha_{\mathcal A}^{\mathcal C};d_{\mathcal A,\cdot j})}{\tilde c_j}$. \\
For DRF,
$\mu_i x_i=t\ \Rightarrow\ x_i=t\alpha_i,\qquad
t=\min_{j\in\mathcal M_0}\frac{c_j}{\sum_{i\in\mathcal N}\alpha_i d_{ij}},$\\
$\sum_{i\in\mathcal N}x_i^{\mathrm{DRF}}
=t\sum_{i\in\mathcal N}\alpha_i
=\min_{j\in\mathcal M_0}\frac{c_j}{M_1(\alpha;d_{\cdot j})}
=\frac{c_{j'}}{M_1(\alpha;d_{\cdot j'})},$

where $j'\in\arg\min_{j\in\mathcal M_0}\frac{M_1(\alpha;d_{\cdot j})}{c_j}$.\\
Therefore,
$\sum_{i\in\mathcal N}x_i^{\mathrm{DDRF}}\ge \sum_{i\in\mathcal N}x_i^{\mathrm{DRF}}
\iff
|\mathcal W|+\frac{\tilde c_{\tilde j^\star}}{M_1(\alpha_{\mathcal A}^{\mathcal C};d_{\mathcal A,\cdot \tilde j^\star})}
\ \ge\
\frac{c_{j'}}{M_1(\alpha;d_{\cdot j'})},$

and the reverse inequality yields $\sum_i x_i^{\mathrm{DDRF}}<\sum_i x_i^{\mathrm{DRF}}$.
\end{proof}

\end{document}